\documentclass[final]{IEEEtran}
\usepackage{amsthm,amssymb,graphicx,multirow,amsmath,color,amsfonts,physics}
\usepackage[update,prepend]{epstopdf}
\usepackage[noadjust]{cite}
\usepackage{tikz}
\usepackage{bbm} 
\usepackage{pdfpages}
\usepackage{balance}
\usepackage{multirow}
\usepackage{comment}
\usepackage{subfigure}
\usepackage{yhmath}

\allowdisplaybreaks 
\allowdisplaybreaks 

\usepackage{float}

\begin{document}
\def\nba{{\mathbf{a}}}
\def\nbb{{\mathbf{b}}}
\def\nbc{{\mathbf{c}}}
\def\nbd{{\mathbf{d}}}
\def\nbe{{\mathbf{e}}}
\def\nbf{{\mathbf{f}}}
\def\nbg{{\mathbf{g}}}
\def\nbh{{\mathbf{h}}}
\def\nbi{{\mathbf{i}}}
\def\nbj{{\mathbf{j}}}
\def\nbk{{\mathbf{k}}}
\def\nbl{{\mathbf{l}}}
\def\nbm{{\mathbf{m}}}
\def\nbn{{\mathbf{n}}}
\def\nbo{{\mathbf{o}}}
\def\nbp{{\mathbf{p}}}
\def\nbq{{\mathbf{q}}}
\def\nbr{{\mathbf{r}}}
\def\nbs{{\mathbf{s}}}
\def\nbt{{\mathbf{t}}}
\def\nbu{{\mathbf{u}}}
\def\nbv{{\mathbf{v}}}
\def\nbw{{\mathbf{w}}}
\def\nbx{{\mathbf{x}}}
\def\nby{{\mathbf{y}}}
\def\nbz{{\mathbf{z}}}
\def\nb0{{\mathbf{0}}}
\def\nb1{{\mathbf{1}}}

\def\nbA{{\mathbf{A}}}
\def\nbB{{\mathbf{B}}}
\def\nbC{{\mathbf{C}}}
\def\nbD{{\mathbf{D}}}
\def\nbE{{\mathbf{E}}}
\def\nbF{{\mathbf{F}}}
\def\nbG{{\mathbf{G}}}
\def\nbH{{\mathbf{H}}}
\def\nbI{{\mathbf{I}}}
\def\nbJ{{\mathbf{J}}}
\def\nbK{{\mathbf{K}}}
\def\nbL{{\mathbf{L}}}
\def\nbM{{\mathbf{M}}}
\def\nbN{{\mathbf{N}}}
\def\nbO{{\mathbf{O}}}
\def\nbP{{\mathbf{P}}}
\def\nbQ{{\mathbf{Q}}}
\def\nbR{{\mathbf{R}}}
\def\nbS{{\mathbf{S}}}
\def\nbT{{\mathbf{T}}}
\def\nbU{{\mathbf{U}}}
\def\nbV{{\mathbf{V}}}
\def\nbW{{\mathbf{W}}}
\def\nbX{{\mathbf{X}}}
\def\nbY{{\mathbf{Y}}}
\def\nbZ{{\mathbf{Z}}}

\def\ncalA{{\mathcal{A}}}
\def\ncalB{{\mathcal{B}}}
\def\ncalC{{\mathcal{C}}}
\def\ncalD{{\mathcal{D}}}
\def\ncalE{{\mathcal{E}}}
\def\ncalF{{\mathcal{F}}}
\def\ncalG{{\mathcal{G}}}
\def\ncalH{{\mathcal{H}}}
\def\ncalI{{\mathcal{I}}}
\def\ncalJ{{\mathcal{J}}}
\def\ncalK{{\mathcal{K}}}
\def\ncalL{{\mathcal{L}}}
\def\ncalM{{\mathcal{M}}}
\def\ncalN{{\mathcal{N}}}
\def\ncalO{{\mathcal{O}}}
\def\ncalP{{\mathcal{P}}}
\def\ncalQ{{\mathcal{Q}}}
\def\ncalR{{\mathcal{R}}}
\def\ncalS{{\mathcal{S}}}
\def\ncalT{{\mathcal{T}}}
\def\ncalU{{\mathcal{U}}}
\def\ncalV{{\mathcal{V}}}
\def\ncalW{{\mathcal{W}}}
\def\ncalX{{\mathcal{X}}}
\def\ncalY{{\mathcal{Y}}}
\def\ncalZ{{\mathcal{Z}}}

\def\nbbA{{\mathbb{A}}}
\def\nbbB{{\mathbb{B}}}
\def\nbbC{{\mathbb{C}}}
\def\nbbD{{\mathbb{D}}}
\def\nbbE{{\mathbb{E}}}
\def\nbbF{{\mathbb{F}}}
\def\nbbG{{\mathbb{G}}}
\def\nbbH{{\mathbb{H}}}
\def\nbbI{{\mathbb{I}}}
\def\nbbJ{{\mathbb{J}}}
\def\nbbK{{\mathbb{K}}}
\def\nbbL{{\mathbb{L}}}
\def\nbbM{{\mathbb{M}}}
\def\nbbN{{\mathbb{N}}}
\def\nbbO{{\mathbb{O}}}
\def\nbbP{{\mathbb{P}}}
\def\nbbQ{{\mathbb{Q}}}
\def\nbbR{{\mathbb{R}}}
\def\nbbS{{\mathbb{S}}}
\def\nbbT{{\mathbb{T}}}
\def\nbbU{{\mathbb{U}}}
\def\nbbV{{\mathbb{V}}}
\def\nbbW{{\mathbb{W}}}
\def\nbbX{{\mathbb{X}}}
\def\nbbY{{\mathbb{Y}}}
\def\nbbZ{{\mathbb{Z}}}

\def\nfrakR{{\mathfrak{R}}}

\def\nrma{{\rm a}}
\def\nrmb{{\rm b}}
\def\nrmc{{\rm c}}
\def\nrmd{{\rm d}}
\def\nrme{{\rm e}}
\def\nrmf{{\rm f}}
\def\nrmg{{\rm g}}
\def\nrmh{{\rm h}}
\def\nrmi{{\rm i}}
\def\nrmj{{\rm j}}
\def\nrmk{{\rm k}}
\def\nrml{{\rm l}}
\def\nrmm{{\rm m}}
\def\nrmn{{\rm n}}
\def\nrmo{{\rm o}}
\def\nrmp{{\rm p}}
\def\nrmq{{\rm q}}
\def\nrmr{{\rm r}}
\def\nrms{{\rm s}}
\def\nrmt{{\rm t}}
\def\nrmu{{\rm u}}
\def\nrmv{{\rm v}}
\def\nrmw{{\rm w}}
\def\nrmx{{\rm x}}
\def\nrmy{{\rm y}}
\def\nrmz{{\rm z}}

\def\nbydef{:=}
\def\nborel{\ncalB(\nbbR)}
\def\nboreld{\ncalB(\nbbR^d)}
\def\sinc{{\rm sinc}}

\newtheorem{lemma}{Lemma}
\newtheorem{thm}{Theorem}
\newtheorem{definition}{Definition}
\newtheorem{ndef}{Definition}
\newtheorem{nrem}{Remark}
\newtheorem{theorem}{Theorem}
\newtheorem{prop}{Proposition}
\newtheorem{cor}{Corollary}
\newtheorem{example}{Example}
\newtheorem{remark}{Remark}
\newtheorem{assumption}{Assumption}
	

\newcommand{\ceil}[1]{\lceil #1\rceil}
\def\argmin{\operatorname{arg~min}}
\def\argmax{\operatorname{arg~max}}
\def\figref#1{Fig.\,\ref{#1}}%
\def\E{\mathbb{E}}
\def\EE{\mathbb{E}^{!o}}
\def\P{\mathbb{P}}
\def\pc{\mathtt{P_c}}
\def\rc{\mathtt{R_c}}   
\def\p{p}

\def\V{\operatorname{Var}}
\def\erfc{\operatorname{erfc}}
\def\erf{\operatorname{erf}}
\def\opt{\mathrm{opt}}
\def\R{\mathbb{R}}
\def\Z{\mathbb{Z}}

\def\LL{\mathcal{L}^{!o}}
\def\var{\operatorname{var}}
\def\supp{\operatorname{supp}}

\def\N{\sigma^2}
\def\T{\beta}							
\def\sinr{\mathtt{SINR}}			
\def\snr{\mathtt{SNR}}
\def\sir{\mathtt{SIR}}
\def\ase{\mathtt{ASE}}
\def\se{\mathtt{SE}}

\def\calN{\mathcal{N}}
\def\FE{\mathcal{F}}
\def\calA{\mathcal{A}}
\def\calK{\mathcal{K}}
\def\calT{\mathcal{T}}
\def\calB{\mathcal{B}}
\def\calE{\mathcal{E}}
\def\calP{\mathcal{P}}
\def\calL{\mathcal{L}}


\def\l{\ell}
\newcommand{\fad}[2]{\ensuremath{\mathtt{h}_{#1}[#2]}}
\newcommand{\h}[1]{\ensuremath{\mathtt{h}_{#1}}}

\newcommand{\err}[1]{\ensuremath{\operatorname{Err}(\eta,#1)}}
\newcommand{\FD}[1]{\ensuremath{|\mathcal{F}_{#1}|}}



\def\Bx{{\mathcal{B}}^x}
\def\Bxx{{\mathcal{B}}^{x_0}}
\def\jx{y}
\def\m{(\bar{n}-1)}
\def\mm{\bar{n}-1}
\def\Nx{{\mathcal{N}}^x}
\def\Nxo{{\mathcal{N}}^{x_0}}
\def\wj{w_{jx_0}}
\def\uij{u_{jx}}
 \def\yj{y}
 \def\yjx{y}
 \def\zjx{z_x}
 \def \tx {y_0}
 \def \htx {h_0}

\def\rx{z_{1}}
\def\ry{z_{2}}

\def\Rx{Z_{1}}
\def\Ry{Z_{2}}

\def \hyxx {h_{y_{x_0}}}
\def \hyx {h_{y_x}}

\def\nbb1{\mathbbm{1}}
\def\xi{\textbf{x}_i}
\def\xj{\textbf{x}_j}
\def\xk{\textbf{x}_k}
\def\xx{\textbf{x}_0}
\def\yk{\textbf{y}_k}
\def\yj{\textbf{y}_j}
\def\yy{\textbf{y}_0}
\def\oe{\textbf{o}_e}
\def\zl{\textbf{z}_l}
\def\wik{\textbf{w}_{i,k}}
\def\ie{{\em i.e. }}
\def\eg{{\em e.g. }}
\def\iid{{\em i.i.d. }}
\def\avg{\rm avg}

\def\rmnuma{\rm\uppercase\expandafter{\romannumeral1}}
\def\rmnumb{\rm\uppercase\expandafter{\romannumeral2}}
\def\rmnumc{\rm\uppercase\expandafter{\romannumeral3}}
\def\rmnumd{\rm\uppercase\expandafter{\romannumeral4}}
\def\rmnume{\rm\uppercase\expandafter{\romannumeral5}}
\def\rmnumf{\rm\uppercase\expandafter{\romannumeral6}}
\pagenumbering{gobble}
\graphicspath{{./Figures/}}
\title{Connectivity of HAPS-based Solutions for\\ Large-scale Wireless Networks:\\ A Percolation Theory Analysis}
\author{
 Hao Lin,~\IEEEmembership{Graduate Student Member,~IEEE},  Mustafa A. Kishk,~\IEEEmembership{Member,~IEEE}\\ and Mohamed-Slim Alouini,~\IEEEmembership{Fellow,~IEEE}
\thanks{Hao Lin is with the Electrical and Computer Engineering Program, CEMSE Division, King Abdullah University of Science and Technology (KAUST),
Thuwal 23955-6900, Saudi Arabia (e-mail: hao.lin.std@gmail.com).\\
\indent Mustafa A. Kishk is with the Department of Electronic Engineering,
Maynooth University, Maynooth, W23 F2H6 Ireland (e-mail:
mustafa.kishk@mu.ie).\\
\indent Mohamed-Slim Alouini is with the CEMSE Division, King Abdullah
University of Science and Technology (KAUST), Thuwal 23955-6900,
Saudi Arabia (e-mail: slim.alouini@kaust.edu.sa).}
}

\maketitle
\vspace{-2cm}
\begin{abstract}
In the era of sixth-generation (6G) wireless communication, numerous applications are expected to be realized, including environmental monitoring, smart agriculture, remote education, security protection, and intelligent transportation systems. These scenarios require large-scale, continuous Internet services in forests, rivers, oceans, and road networks, to name a few, where optical cables are difficult to deploy. High-altitude platform stations (HAPSs) emerge as a promising solution, offering low-latency, high-capacity services while facilitating the establishment of vertical heterogeneous networks (vHetNets) in fiber-less areas. This paper investigates three HAPS-based solutions, where HAPSs can serve wireless devices directly or via gateway (GW) networks: the HAPS-to-device (H2D) scheme, the HAPS-to-GW-to-device (H2G2D) scheme, and the hybrid scheme. Leveraging percolation theory, we study the feasibility of large-scale continuous Internet coverage, where the key performance indicator (KPI) is the percolation probability. We discuss the subcritical and supercritical cases in different coverage schemes, and prove that the phase transition from zero to non-zero percolation probability appears when increasing the HAPS density or GW density. Numerical results verify that the curve of the critical condition of the phase transition exists between the derived lower bound and upper bound, which can help reduce the upfront cost of HAPS-based vHetNet solutions.
\end{abstract}

\section{Introduction} \label{sec:Intro}

With the development of wireless sensing and communication, more and more applications, such as wearable devices, remote education, intelligent transportation systems, and smart agriculture are expected to be realized \cite{9714482,9042251}. Therefore, the traditional terrestrial network (TN) infrastructure faces explosive growth in connection requirements \cite{9369324}. The sixth-generation (6G) mobile network is expected to integrate non-terrestrial networks (NTNs) with TNs to provide seamless, low-latency and high-capacity connectivity \cite{10396843}, where high altitude platform stations (HAPSs) present a promising solution.

HAPSs operate in the stratosphere, approximately from 20 ${\rm km}$ to 50 ${\rm km}$ above ground, and offer their unique advantages \cite{9380673,10355104}. For instance, HAPSs can provide better line-of-sight (LoS) conditions than TN base stations (BSs). Compared with low-altitude platforms (LAPs) and drones, HAPSs can offer greater stability and wider coverage areas. In the vision of ultra-reliable low-latency communication (URLLC), HAPSs can ensure lower latency than satellites, with higher flexibility and stability. Furthermore, advances in lightweight composite materials, solar panel technology, and autonomous flight technology have made HAPSs more economically viable \cite{10355104}. Therefore, HAPSs can play a vital role in vertical heterogeneous networks (vHetNets) and bridge TNs and the core network through the sky \cite{10417095}. They are expected to help connect rural, remote, and post-disaster regions and accelerate the large-scale Internet service coverage. It is required, however, before mass deployment of such technologies, to use analytical techniques to validate the impact of integrating TNs and NTNs on large-scale wireless service connectivity. 

Stochastic geometry is an important mathematical tool for studying the performance of large-scale wireless networks and it can be used to analyze key performance indicators (KPIs) such as delay, coverage probability, and capacity. In the 6G era, new capabilities for large-scale continuous service should be evaluated. As an application of stochastic geometry and graph theory, percolation theory is used to study the feasibility of large-scale multi-hop networks and continuous coverage areas, where the main KPI is the percolation probability. Especially, the phase transition of percolation probability from zero to non-zero indicates the critical condition to ensure the feasibility of large-scale networks, which can reduce the capital expenditure (CAPEX) and operating expenditure (OPEX) for specific coverage requirements. Therefore, this paper aims to use percolation theory to study the connectivity of different Internet coverage schemes based on HAPSs. 

\subsection{Related Work}
In this paper, we consider three HAPS-based solutions, to realize large-scale continuous Internet coverage areas, where we use percolation theory to evaluate the connectivity in different schemes. Therefore, we divide the related work into: (i) stochastic geometry for HAPS-based solutions and (ii) percolation theory applied on wireless networks.

\indent \textit{Stochastic geometry for HAPS-based solutions:} Stochastic geometry is widely used to capture the performance of large-scale wireless networks, without losing accuracy and tractability \cite{haenggi2009stochastic}. In the vision of future HAPS networks, stochastic geometry can help model the spatial locations of nodes and evaluate large-scale HAPS network performance \cite{9380673}. In \cite{10634042}, authors highlighted NTN’s common system-level metrics, including coverage-based, relay-based, and routing-based analysis. Firstly, HAPSs can be integrated with TN to improve total coverage performance. For example, authors in \cite{10082988} evaluated the performance of spectrum sharing between terrestrial BSs and HAPSs, where suitable deployment density of HAPSs can improve coverage probability and transmission capacity. Cooperation between different NTN platforms can make the best of their respective advantages. Authors in \cite{8533584} established an aerial heterogeneous network architecture using HAPSs and LAPs, where a dynamic LAP placement can enhance the quality of service of HAPSs. In \cite{10050345}, authors compared the performance of LAP-assisted and HAP-assisted systems in post-disaster areas, where HAPSs are not recommended for use in case of relatively small disasters, but provide strong transmit power when the considered area is large. In \cite{10097717}, authors designed an integrated HAPS and LAP system to provide backhaul services for the post-disaster areas and compared the performance of HAP-assisted backhaul and LAP direct backhaul schemes. Many works integrated the space and airborne platforms to realize ubiquitous service and provide opportunities for computing tasks. Authors in \cite{9520123} investigated the performance of cache-enabled hybrid satellite-aerial-terrestrial networks, which can reduce transmission latency and provide an expandable framework for evolved networks. In \cite{9841465}, authors considered a cooperative satellite-aerial-terrestrial network, where information from a group of aerial terminals can be forwarded to the terrestrial destination through satellites or multiple aerial platforms. Authors in \cite{10506977} considered the cost constraints and evaluated the percentage of unmanned aerial vehicles (UAVs), HAPSs, and satellites at different total costs. They proposed an algorithm to obtain better total connection probability at the same cost. Furthermore, the performance of vHetNet architectures becomes more and more important. In \cite{10341311}, authors evaluated the joint coverage probability and transmission rate of the space-air-ground integrated networks (SAGIN), which incorporate multi-band terahertz (THz) and radio frequency (RF) resources. They also proposed an annealing algorithm to optimize the channel allocation. In \cite{10040542}, authors proposed space-air-ground-sea integrated networks (SAGSINs). They studied the coverage probability of stations on far-reaching ocean surfaces and verified the applicability of SAGSINs.

\textit{Percolation theory applied on wireless networks}: Percolation theory is widely used to capture the connectivity of large-scale networks, including multi-hop links, detective paths, continuous coverage, security, to name a few \cite{haenggi2009stochastic,haenggi2012stochastic,elsawy2023tutorial}. For example, percolation theory can be used to derive the critical density of camera sensors in clustered three-dimensional (3D) wireless camera sensor networks \cite{wang2019cooperative}. In \cite{zhaikhan2020safeguarding}, authors characterized the critical density of spatial firewalls, which can physically prevent malware epidemics in large-scale wireless networks. For cognitive networks, the coexistence of random primary and secondary networks was proved in \cite{yemini2019simultaneous}. In \cite{RFIoE1}, the authors proposed an energy-as-a-service platform that can provide wireless energy transfer services for Internet of Everything (IoE) devices. They proved the phase transition of percolation probability in such RF-powered IoE networks and derived the relationship between IoE device density and the critical energy station density \cite{RFIoE1}. For cellular networks, percolation theory can help capture the ability to realize large-scale Internet services through infrastructure sharing technique, especially for mobile users and intelligent vehicles \cite{HLISPC1}. With the development of reconfigurable intelligent surfaces (RIS), more topics of percolation analysis have been investigated recently. In \cite{wu2023connectivity}, authors investigated the connectivity of large-scale RIS-assisted integrated access and backhaul (IAB) networks.
In \cite{zhu2023connectivity}, the authors analyzed the connectivity of networks assisted by the directional antenna and transmissive RIS.
Especially, percolation theory has been used to evaluate the connectivity of NTN platforms. Authors in \cite{anjum2019percolation} derived the bounds of the critical node density of homogeneous and heterogeneous wireless balloon networks. They also derived the critical density of UAVs to ensure network coverage through UAV networks \cite{anjum2020coverage}. Based on spherical stochastic geometry, authors in \cite{HLLEOSat1} investigated the sub-critical and super-critical cases for large-scale continuous service on the earth via low earth orbit (LEO) satellites. They derived the critical density for LEO satellites to generate global continuous services, which can help further realize seamless coverage for the whole earth.

\subsection{Contributions}
The contributions of this paper can be summarized as follows:
\begin{itemize}
    \item We discuss three HAPS-based coverage models: (i) the HAPS-to-Device (H2D) coverage scheme, (ii) the HAPS-to-Gateway-to-Device (H2G2D) coverage scheme and (iii) the hybrid coverage scheme, and define the corresponding random graph for each coverage scheme.
    \item We use the percolation probability to capture the ability to generate large-scale continuous Internet coverage in each coverage scheme. Then, we discuss the sub-critical and super-critical cases in each coverage scheme, where the percolation probability is zero or non-zero, respectively.
    \item We prove the existence of the critical condition for phase transition of percolation probability from zero to non-zero in each HAPS-based coverage scheme, and show the properties of the critical HAPS density and gateway (GW) density, whose theoretical curves should be located between the lower bounds and upper bounds we have obtained.
    \item We emphasize how to use the derived critical conditions for phase transition of percolation probability to reduce the CAPEX and OPEX of HAPS-based solutions for large-scale Internet service coverage.
\end{itemize}

The rest of this paper is organized as follows: In Sec. \ref{sec:SysMod}, we introduce the system models for the H2D, H2G2D, and hybrid coverage schemes, respectively. In Sec. \ref{sec:proofofconcept}, we analyze the sub-critical and super-critical cases for each coverage scheme and prove the concept of phase transition. In Sec. \ref{sec:simulation}, we conduct Monte Carlo simulations and compare different coverage schemes. Finally, we conclude this paper in Sec. \ref{sec:conclusion}.

\section{System Model} \label{sec:SysMod}

In this paper, we focus on the ability of generating large-scale continuous coverage areas using percolation theory. Therefore, we first introduce the percolation in discrete model, continuous model and random disk model, respectively. 
\begin{itemize}
    \item As shown in Fig. \ref{fig:Percolation}(a), the whole plane can be divided into many hexagonal faces using a hexagonal lattice. The hexagons that realize the `required functions' are named `open faces' and others are named `closed faces'. When the open faces are connected and form a giant connected component, the percolation on the random graph is realized.
    \item  As shown in Fig. \ref{fig:Percolation}(b), there exists a continuous open area that performs the required function, which is the face percolation in a continuous model. A continuous face percolation can be considered as a discrete face percolation when the side length of each hexagon approaches zero.
    \item As shown in Fig. \ref{fig:Percolation}(c), a special case of continuous face percolation is the random disk model. When the distance between two neighboring disks is less than or equal to the sum of their radii, these two disks are connected. Multiple connected disks can form a connected component where there exists an `open path' from the left to the right or from the bottom to the top.
\end{itemize}

In wireless communications, our aim is to cover any user equipment (UE) that moves or is located on some random paths. Therefore, percolation represents that a random covered path can be realized under sufficient infrastructure deployment. As shown in Fig. \ref{fig:PercoDesert}, an example is that a device in the center of a desert can follow certain paths and maintain Internet service all the way to the edge of the desert. Such a covered path can also avoid the percolation of the `closed areas'. An example is that any device without Internet service in the desert can quickly obtain Internet service as long as it has a stable moving direction, without having to cross the entire desert.
\begin{figure}
    \centering
    \includegraphics[width=1\linewidth]{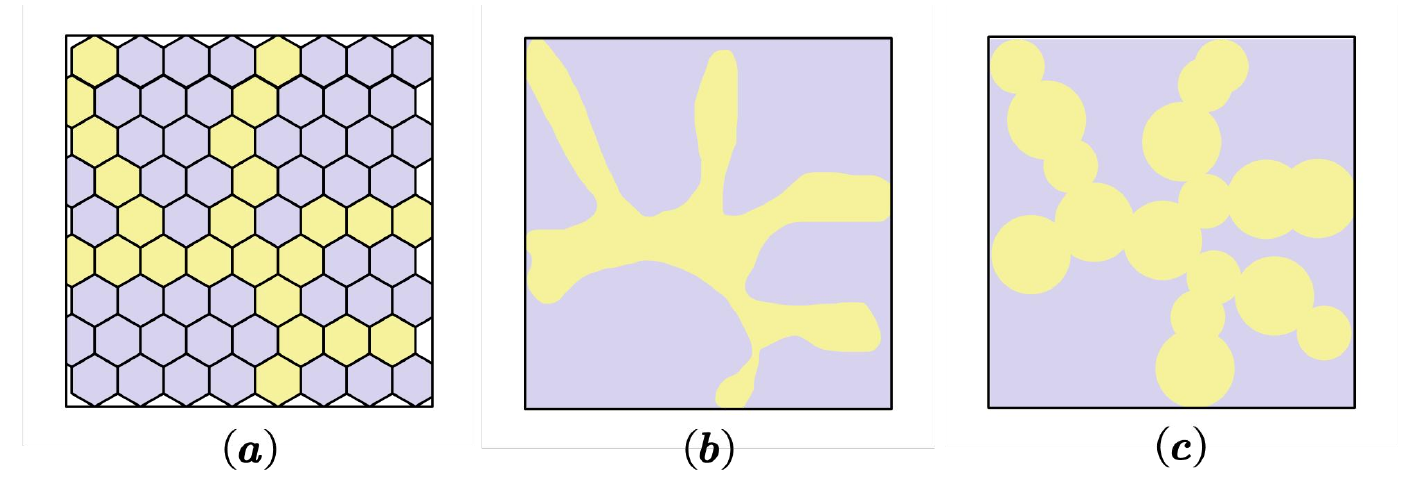}
    \caption{Percolation on (a) discrete model on hexagonal faces, (b) continuous model and (c) random disk model.}
    \label{fig:Percolation}
\end{figure}
\begin{figure}
    \centering
    \includegraphics[width=1\linewidth]{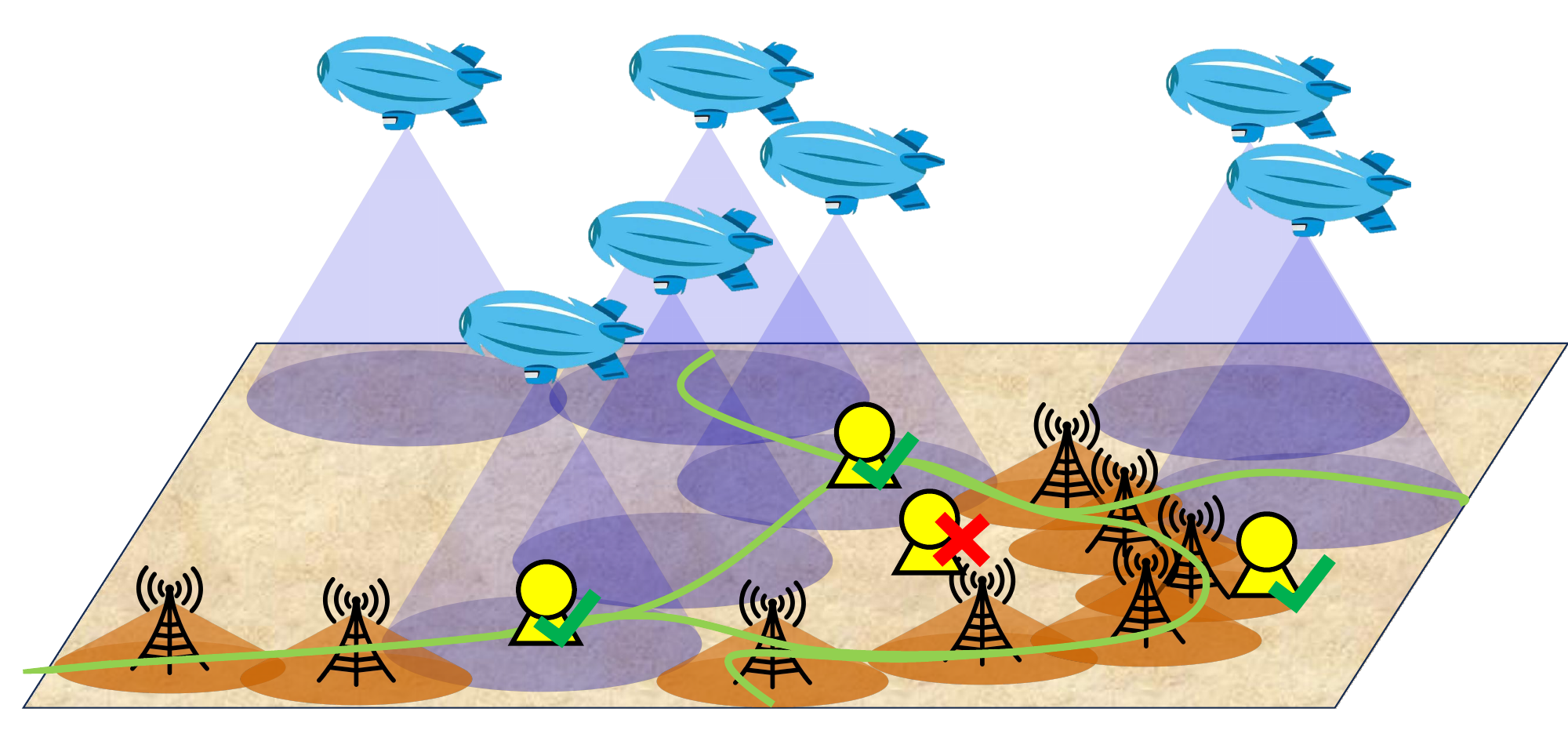}
    \caption{Illustration of percolation in wireless communication. The UE along with the green path can always maintain the Internet service, while the UE without Internet service can find Internet service by searching in any direction.}
    \label{fig:PercoDesert}
\end{figure}

Next, we discuss three HAPS-based solutions for large-scale Internet coverage:
\begin{itemize}
    \item \textbf{HAPS-to-Device Coverage Scheme (H2D)}: In this case, wireless devices can only communicate with HAPSs directly through H2D links. Such a coverage scheme can be considered a single-layer coverage.
    \item \textbf{HAPS-to-GW-to-Device Coverage Scheme (H2G2D)}: In this case, GWs can build a mesh network through GW-to-GW (G2G) links. The wireless devices can not communicate with HAPSs directly, but they can first connect to their neighbor GWs and then communicate with HAPSs indirectly through HAPS-to-GW (H2G) links, G2G links, and GW-to-Device (G2D) links.
    \item \textbf{Hybrid Coverage Scheme}: The wireless devices can form direct links to HAPSs or indirect links via GW mesh networks. Such a coverage scheme can help HAPSs cover more wireless devices than the H2D scheme and H2G2D scheme but requires more spectrum resources.
\end{itemize}

In this paper, the HAPSs we considered are already successfully connected to the core network. It is worth noting that, GWs denote different entities in different realistic heterogeneous networks. For example, in a cellular network supported by HAPSs, GWs refer to cellular BSs and the wireless devices are mobile users. In HAPS-enabled Internet of Things (IoT) networks, GWs mainly serve large numbers of IoT devices, and GWs are called IoT GWs. Furthermore, in certain scenarios, the wireless devices may also include various vehicles or edge computing units. As shown in Fig. \ref{fig:coveragemodels}, we focus on the availability of large-scale continuous Internet coverage through these HAPS-based solutions. Therefore, we respectively introduce the system model for these three coverage schemes.

\begin{figure}[htbp]
\centering
\subfigure[Case 1: HAPS-to-Device (H2D) coverage scheme.]{
\begin{minipage}[t]{1\linewidth}
\centering 
\includegraphics[width=1\textwidth]{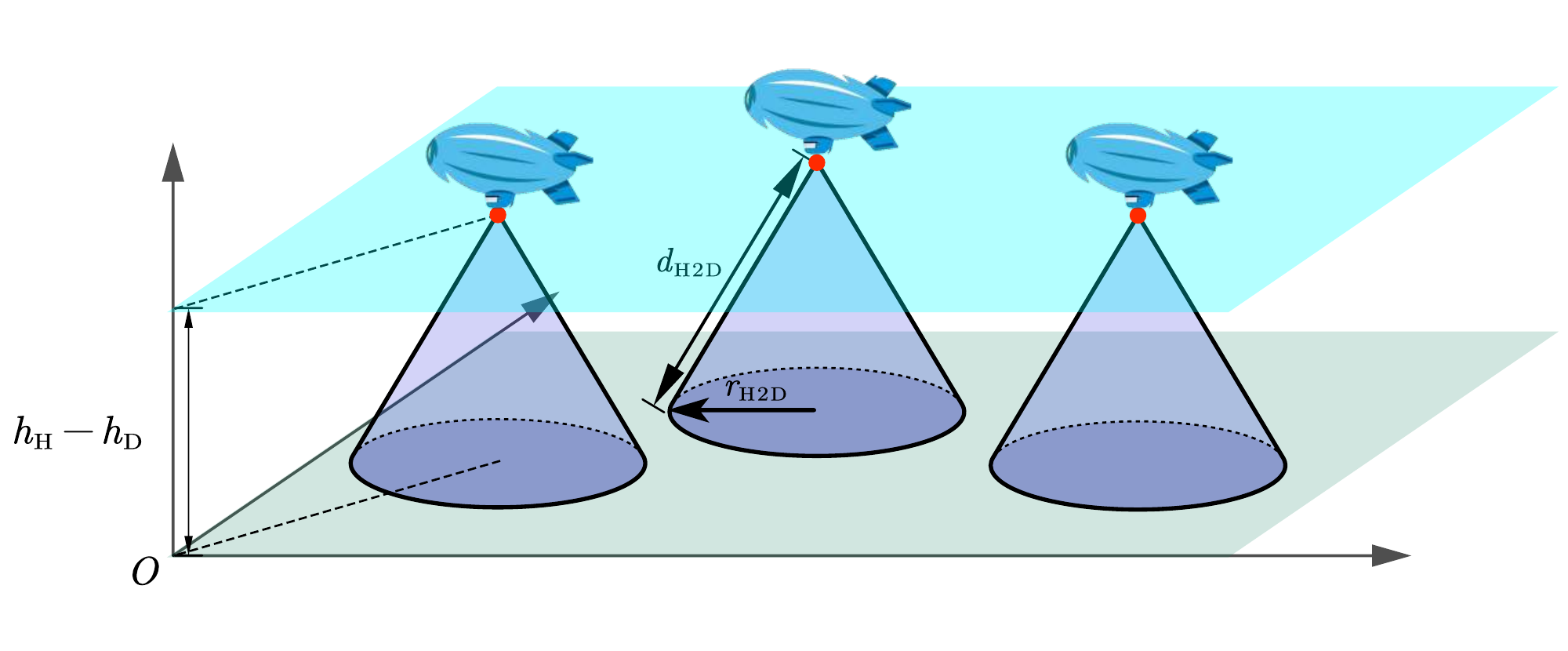}
\label{fig:H2D}
\end{minipage}}

\subfigure[Case 2: HAPS-to-GW-to-Device (H2G2D) coverage scheme.]{
\begin{minipage}[t]{1\linewidth}
\centering 
\includegraphics[width=1\textwidth]{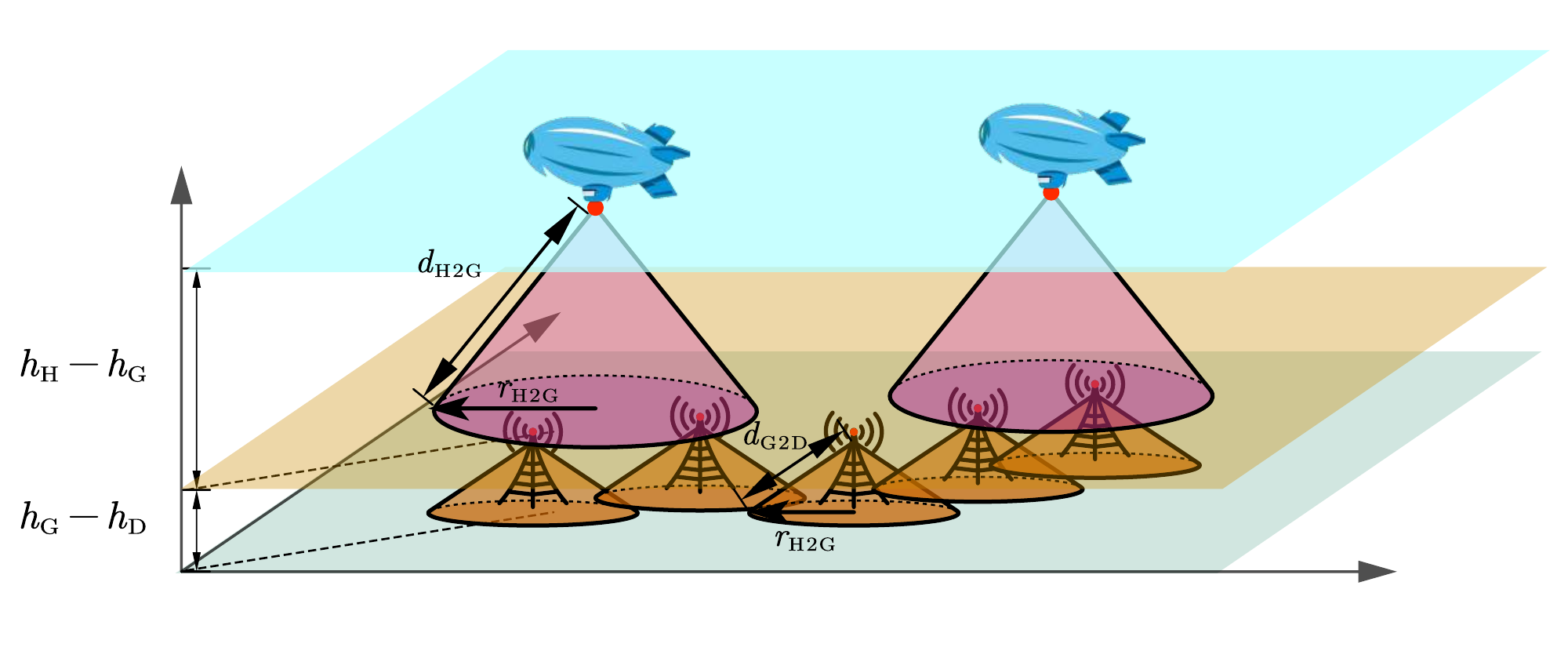}
\label{fig:H2G2D}
\end{minipage}}

\subfigure[Case 3: Hybrid coverage scheme.]{
\begin{minipage}[t]{1\linewidth}
\centering 
\includegraphics[width=1\textwidth]{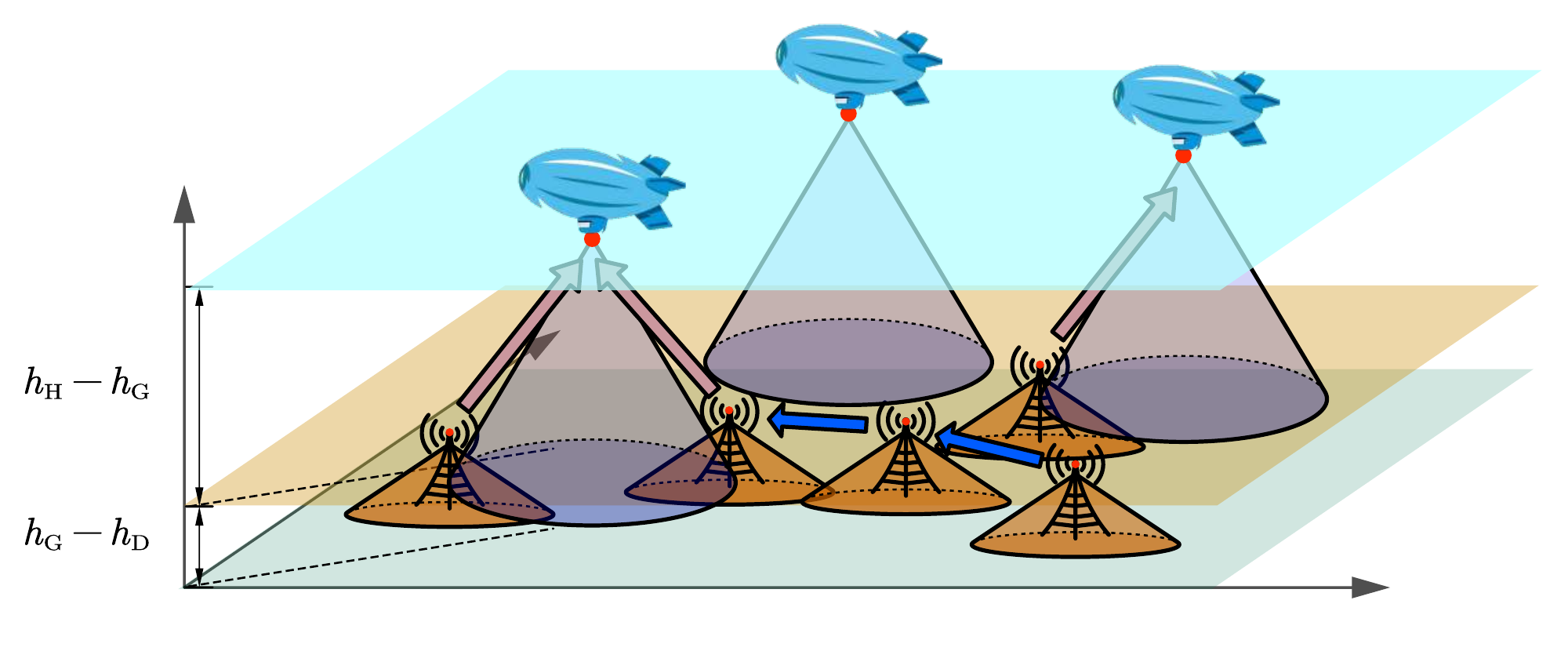}
\label{fig:Hybrid}
\end{minipage}}
\caption{Illustration of H2D, H2G2D and hybrid coverage schemes.}
\label{fig:coveragemodels}
\end{figure}

\subsection{HAPS-to-Device Coverage}
As shown in Fig. \ref{fig:H2D}, we first discuss the case where HAPSs directly communicate with wireless devices on the ground, and investigate the continuity of such H2D coverage areas. We assume that HAPSs are deployed at the same altitude $h_{\rm H}$ and the deployment follows a two-dimensional (2D) Poisson point process (PPP) with density $\lambda_{\rm H}$. Therefore, the set of HAPSs' coverage centers can be modeled as a PPP $\Psi$ with the same density $\lambda_{\rm H}$, where 
\begin{equation}
    \Psi=\{\yj\},
\end{equation}
and $\yj$ represents the location of one HAPS's coverage center on the ground. Each HAPS can provide service for wireless devices within a distance of $d_{\rm H2D}$. We assume that potential wireless devices are on the ground, therefore, the coverage area for wireless devices of the HAPS with projection at $\yj$ can be expressed as:
\begin{equation}
    \mathcal{A}_j=\{\zl\in \mathbb{R}^2:\|\zl-\yj\|\leq r_{\rm H2D}\}
\end{equation}
where
\begin{equation}
    r_{\rm H2D}=\sqrt{d_{\rm H2D}^2-(h_{\rm H}-h_{\rm D})^2},
\end{equation}
$h_{\rm D}$ is the altitude of potential wireless devices and $\|\cdot\|$ means the Euclidean norm.\\
\indent In this paper, we want to investigate the continuity of service areas of HAPSs for wireless devices. Considering only the H2D coverage, we define the random graph $G_{\rm H2D}=\{V_{\rm H2D},E_{\rm H2D}\}$, where the vertex set $V_{\rm H2D}=\Psi$ and the edge set $E_{\rm H2D}$ can be represented as:
\begin{equation}
    E_{\rm H2D}=\{\overline{\yj\yk}:\|\yj-\yk\|\leq 2r_{\rm H2D},\yj,\yk\in V_{\rm H2D}\}
\label{E_H2D}
\end{equation}
because when the distance between $\yj$ and $\yk$ is smaller than $2r_{\rm H2D}$, the coverage areas of these two HAPSs have overlapping areas, that is $\mathcal{A}_j\cap \mathcal{A}_k\neq \varnothing$. Such an H2D coverage can be described as a typical Gilbert disk model.
In this case, percolation probability indicates the probability of generating large-scale continuous service areas through H2D coverage. It represents the probability of generating a giant component containing the origin $K_{\rm H2D}(0)\subseteq G_{\rm H2D}$, whose cardinality is infinite, \ie
\begin{equation}
    \theta_{\rm H2D}(\lambda_{\rm H})=\P\{|K_{\rm H2D}(0)|=\infty\}.
\end{equation}
In this coverage scheme, the percolation probability of the whole system is related to the HAPS density. Therefore, we aim to evaluate the critical condition for the phase transition of percolation probability from zero to non-zero, and the design objective for the HAPS deployment can be formulated as:
\begin{equation}
    \begin{array}{ll}
       \text{ minimize}  & \lambda_{\rm H}  \\
       \text{ subject to}  & \theta_{\rm H2D}(\lambda_{\rm H})>0,\;\lambda_{\rm H}>0.  \\
    \end{array}
    \label{perproH2D}
\end{equation}

\subsection{HAPS-to-GW-to-Device Coverage}
Secondly, we discuss the case where GWs work as relay nodes between HAPSs and wireless devices, as shown in Fig. \ref{fig:H2G2D}. We still assume that HAPSs are deployed at the same altitude $h_{\rm H}$ and their 2D locations follow a PPP $\Psi=\{\yj\}$ with a density $\lambda_{\rm H}$. We assume that each HAPS can cover the GWs within a distance $d_{\rm H2G}$. The distribution of GWs' 2D locations follows a PPP $\Phi=\{\xi\}$ with density $\lambda_{\rm G}$, where the altitude of GWs is $h_{\rm G}$. Therefore, the coverage area for GWs of the HAPS with projection at $\yj$ can be expressed as:
\begin{equation}
    \mathcal{B}_j=\{\zl\in\mathbb{R}^2:\|\zl-\yj\|\leq r_{\rm H2G}\}
\end{equation}
where
\begin{equation}
    r_{\rm H2G}=\sqrt{d_{\rm H2G}^2-(h_{\rm H}-h_{\rm G})^2}.
\end{equation}
Therefore, the set of 2D locations of GWs directly covered by HAPSs can be expressed as
\begin{equation}
    \Phi_d=\{\xi\in\Phi:\xi\in\bigcup_{j:\yj\in\Psi} \mathcal{B}_j\}.
\end{equation}

\indent We assume that there exist wireless mutual communications between GWs, and the maximum GW-to-GW (G2G) communication distance is $r_{\rm G2G}$. Therefore, the GWs outside  the HAPSs' direct coverage can connect to HAPSs through the GW mesh network formed by such multi-hop G2G links. We define the location set of 2D locations of GWs connected to the core network through multi-hop G2G communication as:
\begin{equation}
\begin{array}{r@{}l}
    \Phi_{id}=\{&\xi^0\in\Phi\backslash\Phi_d:\;\exists\xi^1,...,\xi^{N-1}\in \Phi\backslash \Phi_d,\xi^N\in\Phi_d,\\
    &\|\xi^k-\xi^{k+1}\|\leq r_{\rm G2G},\forall k=0,...,N-1,N\geq 1\}.
\end{array}
\end{equation}
Therefore, the set of 2D locations of all GWs that can be successfully connected to the core network is defined as $\Phi_c=\Phi_d \cup \Phi_{id}$.

In this coverage scheme, we aim to investigate the continuity of the coverage of GWs, which are successfully connected to the HAPSs directly or via multi-hop G2G communication. Each GW can build a communication link to the wireless devices within a distance of $d_{\rm G2D}$. Therefore, each GW at $\xi$ can form a circular service area $\mathcal{C}_i$ on the ground, where
\begin{equation}
    \mathcal{C}_i=\{\zl\in\mathbb{R}^2:\|\zl-\xi\|\leq r_{\rm G2D}\}
\end{equation}
and
\begin{equation}
    r_{\rm G2D}=\sqrt{d_{\rm G2D}^2-(h_{\rm G}-h_{\rm D})^2}.
\end{equation}
Considering the H2G and G2D coverage and G2G communications, we define the random graph $G_{\rm H2G2D}=\{V_{\rm H2G2D},E_{\rm H2G2D}\}$, where the vertex set $V_{\rm H2G2D}=\Phi_c$ and the edge set $E_{\rm H2G2D}$ can be expressed as:
\begin{equation}
    E_{\rm H2G2D}=\{\overline{\xi\xk}:\|\xi-\xk\|\leq 2r_{\rm G2D},\xi,\xk\in V_{\rm H2G2D}\}
\label{E_H2G2D}
\end{equation}
because when the distance between $\xi$ and $\xk$ is smaller than $2r_{\rm G2D}$, the coverage areas of these two GWs have overlapping areas, that is $\mathcal{C}_i\cap \mathcal{C}_k\neq \varnothing$. 
In this case, large-scale continuous service areas rely on the H2G, G2D coverage, and G2G communications. Therefore, the percolation probability can be defined as the probability of generating a giant component containing the origin $K_{\rm H2G2D}(0)\subseteq G_{\rm H2G2D}$, whose cardinality is infinite, \ie
\begin{equation}
    \theta_{\rm H2G2D}(\lambda_{\rm H},\lambda_{\rm G})=\P\{|K_{\rm H2G2D}(0)|=\infty\}.
\end{equation}

In the H2G2D coverage scheme, the percolation probability is related to the densities of HAPSs and GWs. Therefore, when we evaluate the critical condition for the phase transition of percolation probability from zero to non-zero, and the design objective can be formulated as two parts:
\begin{equation}
    \begin{array}{ll}
       \text{ minimize}  & \lambda_{\rm H}  \\
       \text{ subject to}  & \theta_{\rm H2G2D}(\lambda_{\rm H},\lambda_{\rm G})>0,  \\
       & \lambda_{\rm H}>0, \; \lambda_{\rm G}>0.  \\
    \end{array}
    \label{perproH2G2DHAPS}
\end{equation}
and 
\begin{equation}
    \begin{array}{ll}
       \text{ minimize}  & \lambda_{\rm G}  \\
       \text{ subject to}  & \theta_{\rm H2G2D}(\lambda_{\rm H},\lambda_{\rm G})>0,  \\
       & \lambda_{\rm H}>0, \; \lambda_{\rm G}>0.  \\
    \end{array}
    \label{perproH2G2DGW}
\end{equation}

\subsection{Hybrid Coverage}
Next, we discuss the hybrid coverage scheme. As shown in Fig. \ref{fig:Hybrid}, HAPSs can not only 
directly cover wireless devices, but also cover the GWs first and then cover the wireless devices. We define the random graph $G_{\rm hbd}=\{V_{\rm hbd},E_{\rm hbd}\}$, where the vertex set $V_{\rm hbd}=V_{\rm H2D}\cup V_{\rm H2G2D}$ and the edge set $E_{\rm hbd}=E_{\rm H2D}\cup E_{\rm H2G2D}\cup E_{\rm HG}$. Especially, $E_{\rm H2D}$, $E_{\rm H2G2D}$ are defined in (\ref{E_H2D}) and (\ref{E_H2G2D}), respectively, and $E_{\rm HG}$ is defined as:
\begin{equation}
\begin{array}{r@{}l}
    E_{\rm HG}=\{\overline{\xi\yj}:\|\xi-\yj\|\leq r_{\rm G2D}&+r_{\rm H2D},\\
    &\yj\in\Psi,\xi\in\Phi_c\}.
\end{array}
\end{equation}
Similarly, the percolation probability in this case can be defined as:
\begin{equation}
    \theta_{\rm hbd}(\lambda_{\rm H},\lambda_{\rm G})=\P\{|K_{\rm hbd}(0)|=\infty\},
\end{equation}
where $K_{\rm hbd}(0)$ represents the giant component containing the origin in the random graph $G_{\rm hbd}$.\\
\indent In the hybrid coverage scheme, the percolation probability is related to the densities of HAPSs and GWs. Therefore, the design objective can be also formulated as two parts:
\begin{equation}
    \begin{array}{ll}
       \text{ minimize}  & \lambda_{\rm H}  \\
       \text{ subject to}  & \theta_{\rm hbd}(\lambda_{\rm H},\lambda_{\rm G})>0,  \\
       & \lambda_{\rm H}>0, \; \lambda_{\rm G}>0.  \\
    \end{array}
    \label{perprohybridHAPS}
\end{equation}
and 
\begin{equation}
    \begin{array}{ll}
       \text{ minimize}  & \lambda_{\rm G}  \\
       \text{ subject to}  & \theta_{\rm hbd}(\lambda_{\rm H},\lambda_{\rm G})>0,  \\
       & \lambda_{\rm H}>0, \; \lambda_{\rm G}>0.  \\
    \end{array}
    \label{perprohybridGW}
\end{equation}
 For ease of reading, we summarize the main notations and parameters in the Table \ref{tab:TableOfNotations}.
\begin{table*}[htbp]
\caption{Table of Notations}
\centering
\begin{center}
\resizebox{\textwidth}{!}{
\renewcommand{\arraystretch}{1}
    \begin{tabular}{ {c} | {l} }
    \hline
        \hline
    \textbf{Notation} & \textbf{Description} \\ \hline
    $\lambda_{\rm H}$; $\Psi=\{\yj\}$ & The density of HAPS; the set of projections of HAPS locations on the ground. \\ \hline
    $\lambda_{\rm G}$; $\Phi=\{\xi\}$ & The density of GWs; the set of projections of GW locations on the ground. \\ \hline
    $d_{\rm H2D}$ & The maximum communication distance between HAPSs and potential wireless devices (H2D).\\ \hline
    $d_{\rm H2G}$ & The maximum communication distance between HAPSs and GWs (H2G). \\ \hline
    $d_{\rm G2D}$ & The maximum communication distance between GWs and wireless devices (G2D). \\ \hline
    $h_{\rm H}$; $h_{\rm G}$; $h_{\rm D}$ & The altitudes of HAPSs, GWs and potential wireless devices, respectively. \\ \hline
    $r_{\rm H2D}$; $r_{\rm H2G}$ & The projected H2D coverage range on the ground; the projected H2G coverage range on the ground. \\ \hline
    $r_{\rm G2D}$; $r_{\rm G2G}$ & The projected G2D coverage range on the ground; the communication range between GWs. \\ \hline
    $G_{\rm H2D}=\{V_{\rm H2D},E_{\rm H2D}\}$ & The random graph of H2D coverage scheme with the vertex set $V_{\rm H2D}$ and edge set $E_{\rm H2D}$. \\ \hline
    $K_{\rm H2D}(0)$ & The giant connected component in H2D coverage scheme where HAPSs cover the origin. \\ \hline
    $G_{\rm H2G2D}=\{V_{\rm H2G2D},E_{\rm H2G2D}\}$ & The random graph of H2D coverage scheme with the vertex set $V_{\rm H2G2D}$ and edge set $E_{\rm H2G2D}$. \\ \hline
    $K_{\rm H2G2D}(0)$ & The giant connected component in H2G2D coverage scheme where GWs cover the origin. \\ \hline
    $G_{\rm hbd}=\{V_{\rm hbd},E_{\rm hbd}\}$ & The random graph of H2D coverage scheme with the vertex set $V_{\rm hbd}$ and edge set $E_{\rm hbd}$. \\ \hline
    $K_{\rm hbd}(0)$ & The giant connected component in hybrid coverage scheme where the origin is covered by GWs or HAPSs directly. \\ \hline
    $\theta_{\rm H2D}$, $\theta_{\rm H2G2D}$, $\theta_{\rm hbd}$ & The percolation probability in H2D, H2G2D, and hybrid coverage scheme, respectively. \\ \hline
    $\mathcal{H}_l$, $\mathcal{T}_l$, $a$ & The designed hexagonal face; the minimum circumscribed circle of $\mathcal{H}_l$, the side length of $\mathcal{H}_l$ (also the radius of $\mathcal{T}_l$). \\ \hline
     \hline
    \end{tabular}
}
\end{center}
\label{tab:TableOfNotations}
\end{table*}

\section{Critical States Analysis} \label{sec:proofofconcept}
\indent To evaluate continuous percolation in a large-scale area, we first divide the whole area into hexagonal faces with side length $a>0$. If any point in a hexagon can not be covered, such a hexagon is considered closed. When the probability of the hexagon $\mathcal{H}_{l}$ being closed is larger than 1/2, there exists a closed circuit around the origin \cite{zhaikhan2020safeguarding}. When a typical device moves through the circuit, its service will be interrupted. Therefore, the percolation probability is zero because any point can not generate a connected component whose cardinality is large. This can be described as:
\begin{equation}
    \theta=0\;{\rm for}\; \P\{\mathcal{H}_l\;{\rm is\;closed}\}>\frac{1}{2}.
\end{equation}
For ease of obtaining the lower bound of the critical density of HAPS or GWs, it is feasible to consider the corresponding minimum circumscribed circles (MCCs) $\mathcal{T}_l$'s of the hexagons $\mathcal{H}_l$'s, which is the green circle in Fig. \ref{fig:H2DClosed}. If any point in $\mathcal{T}_l$ can not be covered, the corresponding hexagon $\mathcal{H}_l$ can not be covered at the same time. Therefore, the probability of $\mathcal{H}_l$ being closed and the probability of $\mathcal{T}_l$ being not covered satisfy:
\begin{equation}
    \P\{\mathcal{H}_l\;{\rm is\;closed}\}>\P\{\mathcal{T}_l\;{\rm is\;not\;covered}\}.
\end{equation}
Therefore, 
\begin{equation}
    \P\{\mathcal{T}_l\;{\rm is\;not\;covered}\}>\frac{1}{2}
\label{Tlclosed}
\end{equation}
is a sufficient condition for zero percolation probability. 

Differently, if the hexagon $\mathcal{H}_{l}$ is covered by the Internet service, it is considered to be open. If the probability of each hexagon being open is larger than 1/2, the percolation probability is non-zero because there exists a giant connected component containing the considered point \cite{zhaikhan2020safeguarding}. This can be described as:
\begin{equation}
    \theta>0\;{\rm for}\; \P\{\mathcal{H}_l\;{\rm is\;open}\}>\frac{1}{2}.
\end{equation}
If everywhere in $\mathcal{T}_l$ can be covered, everywhere in the corresponding hexagon $\mathcal{H}_l$ can be covered at the same time. Therefore, the probability of $\mathcal{H}_l$ being open and the probability of $\mathcal{T}_l$ being covered satisfy:
\begin{equation}
    \P\{\mathcal{H}_l\;{\rm is\;open}\}>\P\{\mathcal{T}_l\;{\rm is\;covered}\}.
\end{equation}
Therefore, 
\begin{equation}
    \P\{\mathcal{T}_l\;{\rm is\;covered}\}>\frac{1}{2}
\label{Tlopen}
\end{equation}
is a sufficient condition for non-zero percolation probability.

Based on these, we discuss the sub-critical and super-critical cases and the proof of concept of each coverage scheme.

\subsection{HAPS-to-Device Coverage}
In the H2D coverage scheme, we first discuss the sub-critical case where the percolation probability is zero.\\
\begin{figure}[ht]
    \centering
    \includegraphics[width=0.7\linewidth]{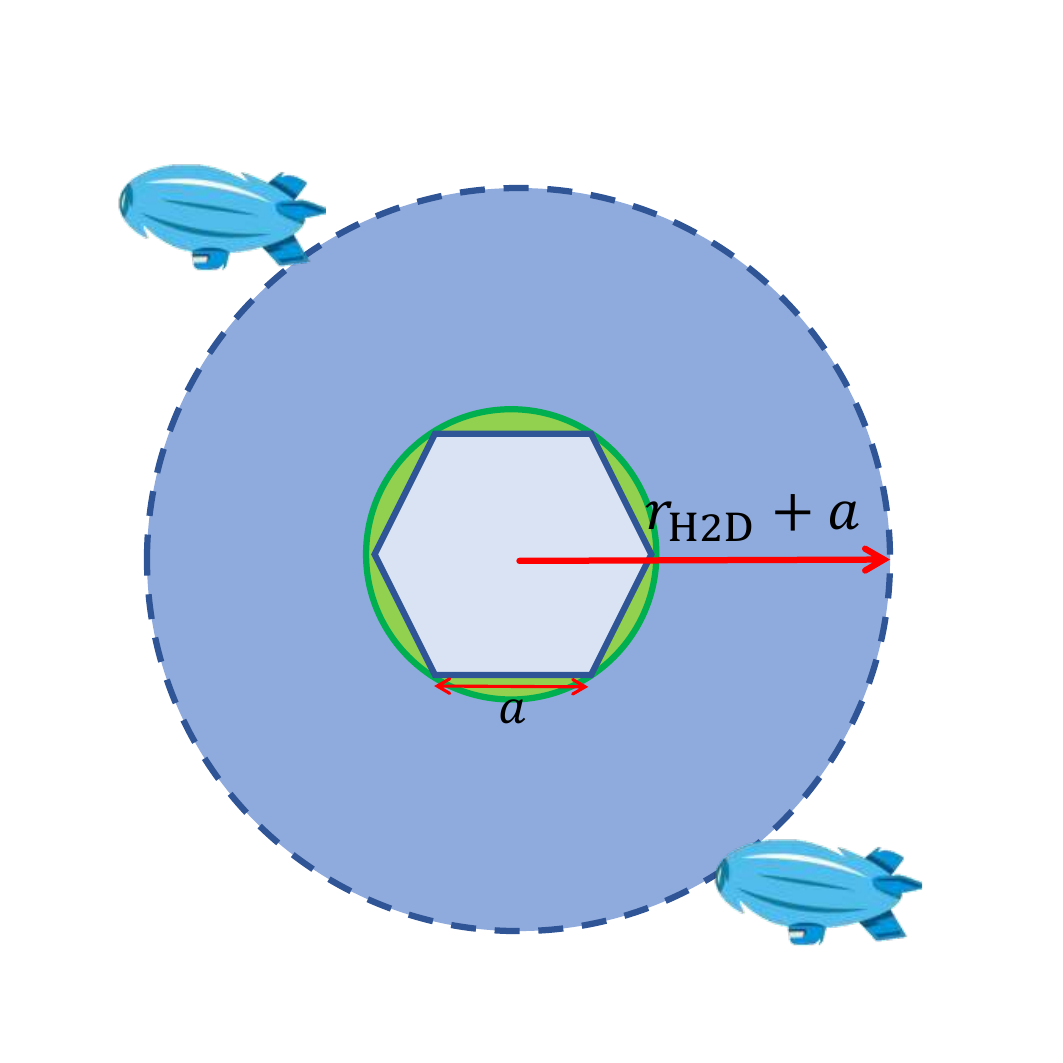}
    \caption{A closed hexagonal face $\mathcal{H}_l$ in the H2D coverage scheme. The HAPSs can not cover the MCC $\mathcal{T}_l$ so that the face $\mathcal{H}_l$ is closed.}
    \label{fig:H2DClosed}
\end{figure}

\textbf{Sub-critical Case:} Consider the MCC $\mathcal{T}_l$ of a hexagon $\mathcal{H}_l$. The radius of $\mathcal{T}_l$ is the same as the side length of the hexagon $a$, and the center of $\mathcal{T}_l$ is the same as the center of $\mathcal{H}_l$. If there is no HAPS inside the circular area with radius $r_{\rm H2D}+a$, any device inside $\mathcal{T}_l$ can not be covered. Therefore, we can obtain:
\begin{equation}
    \P\{\mathcal{T}_l\;{\rm is\;not\;covered}\}=e^{-\lambda_{\rm H}\pi (r_{\rm H2D}+a)^2}.
\label{Tlclosedformula}
\end{equation}
Based on this, we introduce the sufficient condition of HAPS density for no percolation in Theorem \ref{theo:H2Dlower}.
\begin{theorem}
    When the HAPS density $\lambda_{\rm H}$ satisfies $\lambda_{\rm H}<\lambda_{\rm H,L}^{\rm H2D}$, where:
\begin{equation}
    \lambda_{\rm H,L}^{\rm H2D}=\lambda_{\rm H}^{-}\overset{\triangle}{=}\frac{\ln 2}{\pi (r_{\rm H2D}+a)^2},
\end{equation}
the percolation probability of the H2D coverage scheme is zero, \ie $\theta_{\rm H2D}(\lambda_{\rm H})=0$.
\label{theo:H2Dlower}
\end{theorem}
\begin{IEEEproof}
    Substituting (\ref{Tlclosedformula}) into (\ref{Tlclosed}), we can obtain the sufficient condition for zero percolation probability in the H2D coverage scheme.
\end{IEEEproof}
Therefore, when the HAPS density is less than the lower bound $\lambda_{\rm H,L}^{\rm H2D}$, the percolation probability in a large-scale network is zero.
Next, we introduce the super-critical case where the percolation probability is non-zero.

\begin{figure}[ht]
    \centering
    \includegraphics[width=0.7\linewidth]{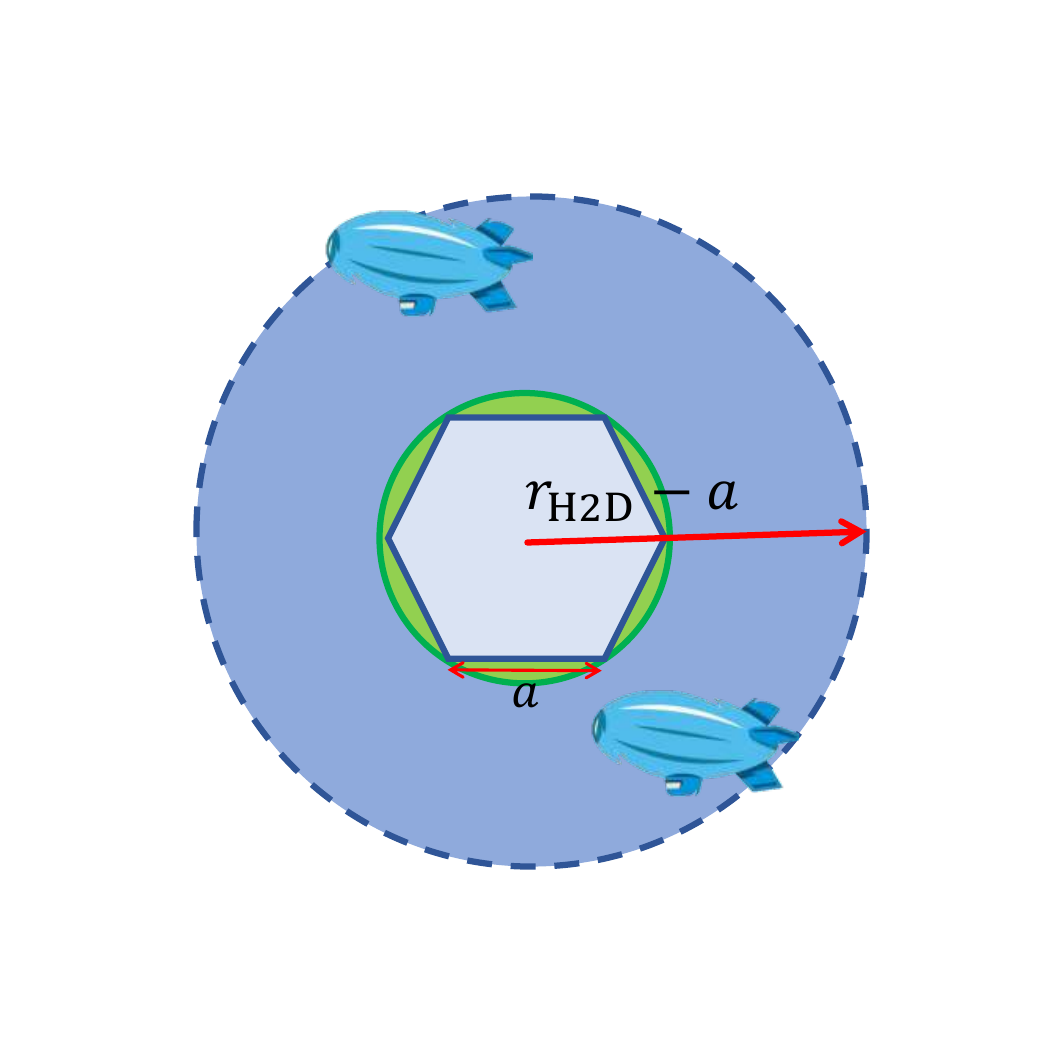}
    \caption{An open hexagonal face $\mathcal{H}_l$ in the H2D coverage scheme. The HAPSs can cover the MCC $\mathcal{T}_l$ so that the face $\mathcal{H}_l$ is open.}
    \label{fig:H2DOpen}
\end{figure}
\textbf{Super-critical Case:} Consider the MCC $\mathcal{T}
_l$ of the hexagon $\mathcal{H}_l$. If there is at least one HAPS inside the circular area with radius $r_{\rm H2D}-a$, any device inside $\mathcal{T}_l$ can be covered. Therefore we can obtain:
\begin{equation}
    \P\{\mathcal{T}_l\;{\rm is\;covered}\}=e^{-\lambda_{\rm H}\pi (r_{\rm H2D}-a)^2}.
\label{Tlopenformula}
\end{equation}
Based on this, we introduce the sufficient condition of HAPS density for non-zero percolation probability in Theorem \ref{theo:H2Dupper}.
\begin{theorem}
    When the HAPS density $\lambda_{\rm H}$ satisfies $\lambda_{\rm H}>\lambda_{\rm H,U}^{\rm H2D}$, where:
\begin{equation}
    \lambda_{\rm H,U}^{\rm H2D}=\lambda_{\rm H}^{+}\overset{\triangle}{=}\frac{\ln 2}{\pi (r_{\rm H2D}-a)^2},
\end{equation}
the percolation probability of the H2D coverage scheme is non-zero, \ie $\theta_{\rm H2D}(\lambda_{\rm H})>0$.
\label{theo:H2Dupper}
\end{theorem}
\begin{IEEEproof}
    Substituting (\ref{Tlopenformula}) into (\ref{Tlopen}), we can obtain the sufficient condition for non-zero percolation probability in the H2D coverage scheme.
\end{IEEEproof}
Therefore, when the HAPS density is larger than the upper bound $\lambda_{\rm H,U}^{\rm H2D}$, the percolation probability in a large-scale network is non-zero. Next, to prove that the critical condition for phase transition of percolation probability exists, we introduce the relationship between the HAPS density and the percolation probability in Lemma \ref{lem:increaseH2D}.
\begin{lemma}
    When the density of HAPSs $\lambda_{\rm H}$ increases, the percolation probability $\theta_{\rm H2D}$ does not decrease, that is
\begin{equation}
    \theta_{\rm H2D}(\lambda_{\rm H,2})\geq\theta_{\rm H2D}(\lambda_{\rm H,1})\;{\rm if}\;\lambda_{\rm H,2}>\lambda_{\rm H,1}.
\end{equation}
\label{lem:increaseH2D}
\end{lemma}
\begin{IEEEproof}
    See Appendix~\ref{app:increaseH2D}.
\end{IEEEproof}
Therefore, with the derived lower bound and upper bound, we can prove the existence of the critical condition of HAPS density for phase transition of percolation probability, which is shown in Lemma \ref{lem:criticalH2D}.

\begin{lemma}
For the H2D coverage scheme, there exists a critical value of HAPS density, that is $\lambda_{\rm H,c}^{\rm H2D}$ where
\begin{equation}
\begin{array}{ll}
    \theta_{\rm H2D}(\lambda_{\rm H})=0 & {\rm for}\;\lambda_{\rm H}\leq \lambda_{\rm H,c}^{\rm H2D}, \\
    \theta_{\rm H2D}(\lambda_{\rm H})>0 & {\rm for}\;\lambda_{\rm H}> \lambda_{\rm H,c}^{\rm H2D}. 
\end{array}
\end{equation}
Especially, the lower bound of $\lambda_{\rm H,c}^{\rm H2D}$ is
\begin{equation}
    \lambda_{\rm H,L}^{\rm H2D}=\lambda_{\rm H}^{-}
\end{equation}
and the upper bound of $\lambda_{\rm H,c}^{\rm H2D}$ is
\begin{equation}
    \lambda_{\rm H,U}^{\rm H2D}=\lambda_{\rm H}^{+}.
\end{equation}
Therefore, we have $\lambda_{\rm H,L}^{\rm H2D}\leq\lambda_{\rm H,c}^{\rm H2D}\leq\lambda_{\rm H,U}^{\rm H2D}$.
\label{lem:criticalH2D}
\end{lemma}
\begin{IEEEproof}
    Because $\theta_{\rm H2D}(\lambda_{\rm H})$ is a non-decreasing function of $\lambda_{\rm H}$, and $\theta_{\rm H2D}(\lambda_{\rm H})=0$ when $\lambda_{\rm H}<\lambda_{\rm H,L}^{\rm H2D}$ and $\theta_{\rm H2D}(\lambda_{\rm H})>0$ when $\lambda_{\rm H}>\lambda_{\rm H,U}^{\rm H2D}$, the critical value for phase transition of percolation probability exists between the derived lower bound and upper bound.
\end{IEEEproof}

\begin{figure}[ht]
    \centering
    \includegraphics[width=1\linewidth]{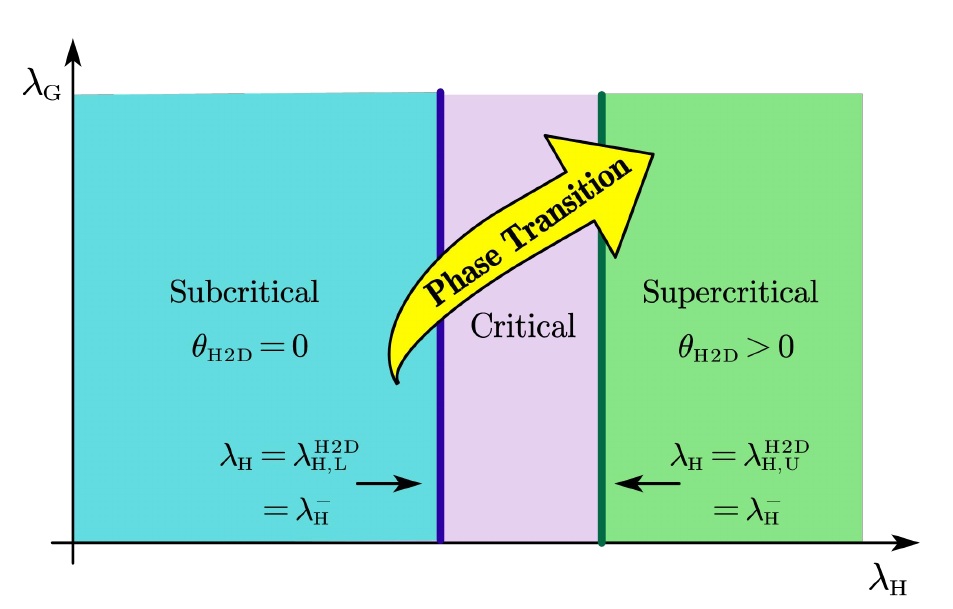}
    \caption{Sub-critical, super-critical, and critical regions for the H2D coverage scheme.}
    \label{fig:criticalregionH2D}
\end{figure}
Notice that the lower bound $\lambda_{\rm H,L}^{\rm H2D}$ and upper bound $\lambda_{\rm H,U}^{\rm H2D}$ are both related to the side length of the designed hexagonal face, we can find out a closed form value of HAPS critical density in the below corollary, which is always located between the lower bound and upper bound.
\begin{cor}
    When the side length of the designed hexagonal faces approaches zero $a\rightarrow 0$, $\lim\limits_{a\rightarrow 0} \lambda_{\rm H,L}^{\rm H2D}=\lim\limits_{a\rightarrow 0} \lambda_{\rm H,U}^{\rm H2D}=\frac{\ln 2}{\pi r_{\rm H2D}^2}$. Therefore, the critical density $\lambda_{\rm H,c}^{\rm H2D}=\frac{\ln 2}{\pi r_{\rm H2D}^2}$ when we consider continuous percolation and the designed hexagons are much smaller then the coverage area of HAPSs.
\end{cor}
However, such a closed-form expression is practical when the considered network is much larger than the coverage radius and the cardinality of giant connected components can be infinite. For finite networks, such a closed-form expression corresponds to a low level of percolation probability, where the percolation probability rapidly increases when the HAPS density grows from this value.

\subsection{HAPS-to-GW-to-Device Coverage}
In the H2G2D coverage scheme, we also introduce the sub-critical case first, where the percolation probability is zero.\\

\begin{figure}[ht]
    \centering
    \includegraphics[width=0.7\linewidth]{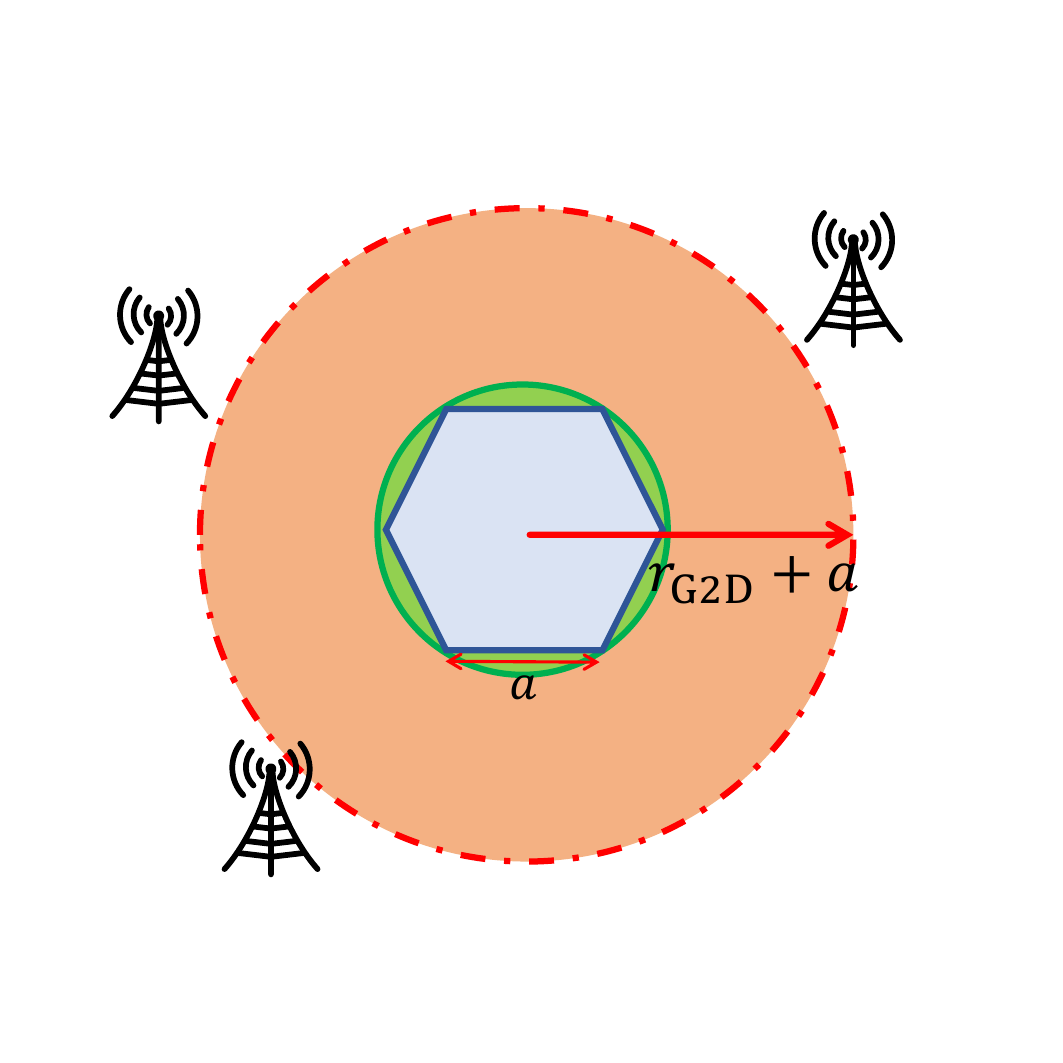}
    \caption{A closed hexagonal face $\mathcal{H}_l$ in the H2G2D coverage scheme. The GWs can not cover the MCC $\mathcal{T}_l$ so that the face $\mathcal{H}_l$ is closed.}
    \label{fig:H2G2DClosed}
\end{figure}

\textbf{Sub-critical Case:} Consider the MCC $\mathcal{T}_l$ of a hexagon $\mathcal{H}_l$. The radius of $\mathcal{T}_l$ is the same as the side length of $\mathcal{H}_l$. If there is no GWs inside the circular area with radius $r_{\rm G2D}+a$, any device inside $\mathcal{T}_l$ can not be covered. Therefore, we can obtain a lower bound of $\P\{\mathcal{T}_l\;{\rm is\;not\;covered}\}$:
\begin{equation}
    \P\{\mathcal{T}_l\;{\rm is\;not\;covered}\}\geq e^{-\lambda_{\rm G}(r_{\rm G2D}+a)^2}.
\label{TlclosedH2G2D}
\end{equation}
Based on this, we introduce the sufficient condition for zero percolation probability in Theorem \ref{theo:H2G2Dlower}.
\begin{theorem}
    When the GW density $\lambda_{\rm G}$ satisfies $\lambda_{\rm G}<\lambda_{\rm G,L}^{\rm H2G2D}(\lambda_{\rm H})$, where:
\begin{equation}
    \lambda_{\rm G,L}^{\rm H2G2D}(\lambda_{\rm H})=\lambda_{\rm G}^{-}\overset{\triangle}{=}\frac{\ln 2}{\pi(r_{\rm G2D}+a)^2},
\label{H2G2Dlowercond}
\end{equation}
the percolation probability of the H2G2D coverage scheme is zero, \ie $\theta_{\rm H2G2D}(\lambda_{\rm H},\lambda_{\rm G})=0$.
\label{theo:H2G2Dlower}
\end{theorem}
\begin{IEEEproof}
    Substituting the lower bound in (\ref{TlclosedH2G2D}) into (\ref{Tlclosed}), we can obtain the sufficient condition for zero percolation probability in the H2G2D coverage scheme.
\end{IEEEproof}
    Especially, the random graph $G_{\rm H2G2D}$ is a subset of the random graph $G=\{V,E\}$, where $V=\Phi$ contains all GWs and $E$ represents whether the GWs in $V$ have coverage overlapping. The sufficient condition for zero percolation probability in $G$ is (\ref{H2G2Dlowercond}), which is also the sufficient condition for no percolation in $G_{\rm H2G2D}$.

Next, we introduce the super-critical case where the percolation probability is non-zero.

\begin{figure}[ht]
    \centering
    \includegraphics[width=0.7\linewidth]{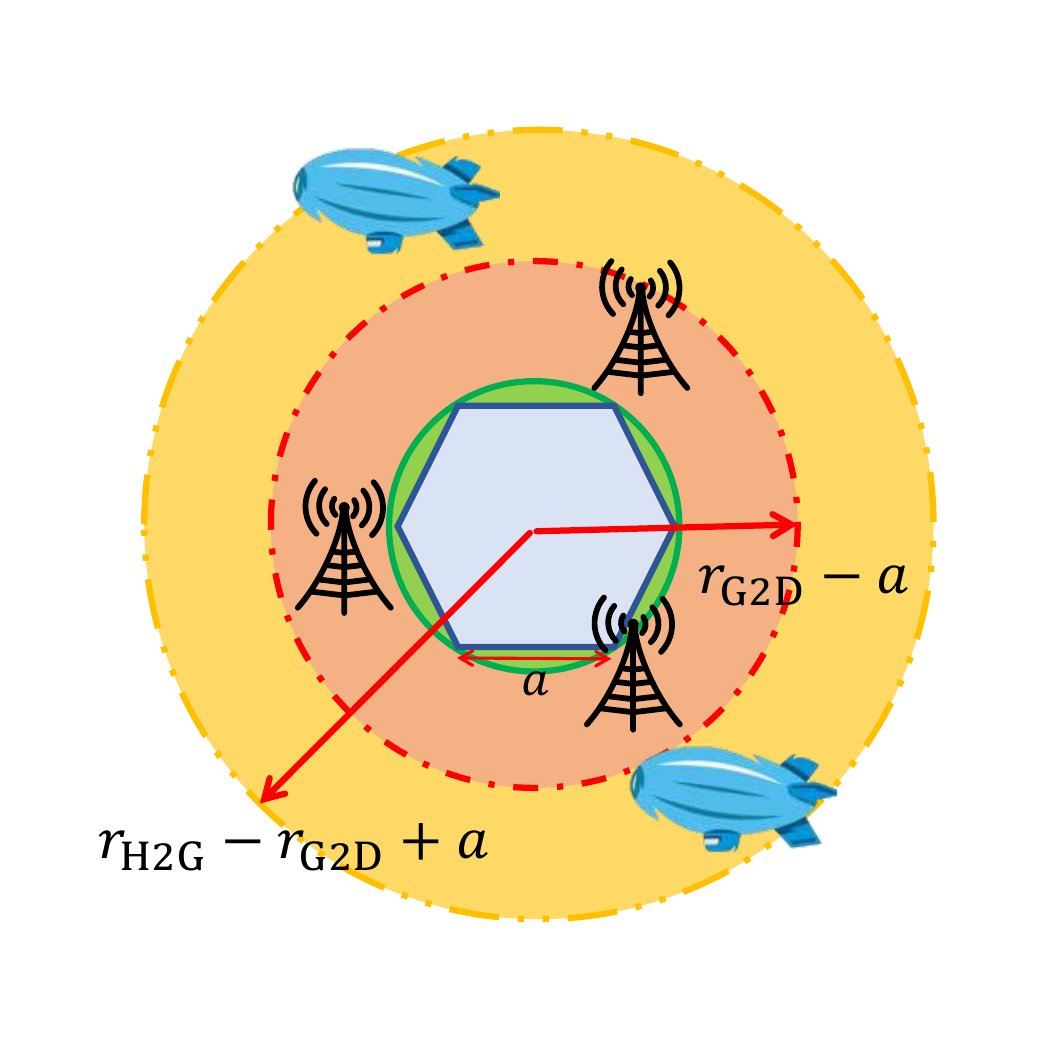}
    \caption{An open hexagonal face $\mathcal{H}_l$ in the H2G2D coverage scheme. The HAPSs can cover the GWs in the designed area, and the GWs in the designed area can cover the MCC $\mathcal{T}_l$ so that the face $\mathcal{H}_l$ is open.}
    \label{fig:H2G2DOpen}
\end{figure}
\textbf{Super-critical Case:} Consider the MCC $\mathcal{T}
_l$ of the hexagon $\mathcal{H}_l$. If there exists at least one GW inside the circular area with radius $r_{\rm G2D}-a$ and at least one HAPS inside the circular area with radius $r_{\rm H2G}-r_{\rm G2D}+a$, any device inside $\mathcal{T}_l$ can be covered. Therefore, we can obtain a lower bound of $\P\{\mathcal{T}_l\;{\rm is\;covered}\}$:
\begin{equation}
\begin{array}{r@{}l}
    \P&\{\mathcal{T}_l\;{\rm is\;covered}\}\\
    &\geq (1-e^{-\lambda_{\rm G}\pi (r_{\rm G2D}-a)^2})(1-e^{-\lambda_{\rm H}\pi (r_{\rm H2G}-r_{\rm G2D}+a)^2}).
\end{array}
\label{TlopenH2G2D}
\end{equation}

Based on this, we obtain the sufficient condition for non-zero percolation probability in Theorem \ref{theo:H2G2Dupper}.
\begin{theorem}
For a fixed value of HAPS density $\lambda_{\rm H}$, when the GW density $\lambda_{\rm G}$ satisfies $\lambda_{\rm G}>\lambda_{\rm G,U}^{\rm H2G2D}(\lambda_{\rm H})$, where:
\begin{equation}
    \lambda_{\rm G,U}^{\rm H2G2D}(\lambda_{\rm H})=\frac{\displaystyle\ln (1+\frac{\frac{1}{2}}{\frac{1}{2}-e^{-\lambda_{\rm H}\pi(r_{\rm H2G}-r_{\rm G2W}+a)^2}})}{\pi(r_{\rm G2D}-a)^2}
\end{equation}
where
\begin{equation}
    \lambda_{\rm H}>\lambda_{\rm H}^{*}
\end{equation}
and
\begin{equation}
    \lambda_{\rm H}^{*}\overset{\triangle}{=}\frac{\ln 2}{\pi (r_{\rm H2G}-r_{\rm G2D}+a)^2},
\end{equation}
percolation probability is non-zero, \ie $\theta_{\rm H2G2D}(\lambda_{\rm H},\lambda_{\rm G})>0$.\\
\indent For a fixed value of GW density $\lambda_{\rm G}>\lambda_{\rm G}^{+}$, where
\begin{equation}
    \lambda_{\rm G}^{+}\overset{\triangle}{=}\frac{\ln 2}{\pi(r_{\rm G2D}-a)^2},
\end{equation}
the sufficient condition for $\theta_{\rm H2G2D}(\lambda_{\rm H},\lambda_{\rm G})>0$ can be rewritten as: 

\begin{equation}
    \lambda_{\rm H}>{\lambda_{\rm G,U}^{\rm H2G2D}}^{-1}(\lambda_{\rm G})\overset{\triangle}{=}\frac{\displaystyle\ln (1+\frac{\frac{1}{2}}{\frac{1}{2}-e^{-\lambda_{\rm G}\pi(r_{\rm H2G}-a)^2}})}{\pi(r_{\rm H2G}-r_{\rm G2D}+a)^2}.
\end{equation}

\label{theo:H2G2Dupper}
\end{theorem}
\begin{IEEEproof}
    Substituting the lower bound in (\ref{TlopenH2G2D}) into (\ref{Tlopen}), we can obtain the sufficient condition for non-zero percolation probability in the H2G2D coverage scheme.
\end{IEEEproof}
    Especially, the super-critical case corresponds to a random graph $G_d'=\{V_d',E_d'\}$, where $V_d'\subseteq\Phi_d$, $\Phi_d\subseteq V_{\rm H2G2D}$ and $E_d'$ represents whether the GWs in $V_d'$ have coverage overlapping. Therefore, $G_d'$ is a subset of $G_{\rm H2G2D}$, and the sufficient condition for non-zero percolation probability in $G_d'$ is (\ref{H2G2Dlowercond}), which is also the sufficient condition for non-zero percolation probability in $G_{\rm H2G2D}$. When $r_{\rm H2G}$ approaches infinity, all GWs can be covered by HAPSs. In this case, the sufficient condition for non-zero percolation probability is only related to the GW density.
\begin{remark}
    When $r_{\rm H2G}\rightarrow \infty$, the sufficient condition for non-zero percolation probability in the H2G2D coverage scheme approaches:
\begin{equation}
    \lambda_{\rm G}>\lambda_{\rm G}^{+},
\end{equation}
where $\lambda_{\rm H}>0$. This is also the upper bound for the simple Gilbert disk model of GWs, because all GWs can be covered by HAPS directly. Also, $\lambda_{\rm G}^{+}$ is the lower bound of $\lambda_{\rm G,L}^{\rm H2G2D}(\lambda_{\rm H})$.
    
\end{remark}
To prove that the critical condition for phase transition of percolation probability exists, we introduce the relationship between HAPS density, GW density and percolation probability in Lemma \ref{lem:increaseH2G2D}.
\begin{lemma}
For a fixed value of HAPS density $\lambda_{\rm H,1}$, the percolation probability $\theta_{\rm H2G2D}$ does not decrease when the GW density increases, that is
\begin{equation}
    \theta_{\rm H2G2D}(\lambda_{\rm H,1},\lambda_{\rm G,2})\geq \theta_{\rm H2G2D}(\lambda_{\rm H,1},\lambda_{\rm G,1})\;{\rm if}\;\lambda_{\rm G,2}>\lambda_{\rm G,1}.
\end{equation}
For a fixed value of GW density $\lambda_{\rm G,1}$, the percolation probability $\theta_{\rm H2G2D}$ does not decrease when the HAPS density increases, that is
\begin{equation}
    \theta_{\rm H2G2D}(\lambda_{\rm H,2},\lambda_{\rm G,1})\geq \theta_{\rm H2G2D}(\lambda_{\rm H,1},\lambda_{\rm G,1})\;{\rm if}\;\lambda_{\rm H,2}>\lambda_{\rm H,1}.
\end{equation}

\label{lem:increaseH2G2D}
\end{lemma}
\begin{IEEEproof}
    See Appendix~\ref{app:increaseH2G2D}.
\end{IEEEproof}

Therefore, with the derived sub-critical and super-critical regions, we can prove the existence of the critical condition of HAPS density and GW density for phase transition of percolation probability in the H2G2D coverage scheme, which is shown in Lemma \ref{lem:criticalH2G2D}.

\begin{lemma}For the H2G2D coverage scheme, as shown in Fig. \ref{fig:criticalregionH2G2D}, the critical condition for phase transition of percolation probability can be described as an implicit function $\mathcal{F}(\lambda_{\rm H},\lambda_{\rm G})=0$, whose curve is between the super-critical region and sub-critical region. It has four properties: (i) the increase in $\lambda_{\rm H}$ does not lead to a higher critical value of $\lambda_{\rm G}$, (ii) the increase in $\lambda_{\rm G}$ does not lead to a higher critical value of $\lambda_{\rm H}$, (iii) for $\lambda_{\rm H}>\lambda_{\rm H}^{*}$, the critical value of $\lambda_{\rm G}$ is between $\lambda_{\rm G,L}^{\rm H2G2D}(\lambda_{\rm H})$ and $\lambda_{\rm G,U}^{\rm H2G2D}(\lambda_{\rm H})$, (iv) for $\lambda_{\rm G}>\lambda_{\rm G}^{+}$, the critical value of $\lambda_{\rm H}$ is between 0 and  ${\lambda_{\rm G,U}^{\rm H2G2D}}^{-1}(\lambda_{\rm G})$.
\label{lem:criticalH2G2D}
\end{lemma}
\begin{IEEEproof}
    See Appendix~\ref{app:criticalH2G2D}.
\end{IEEEproof}
Even though the expression of $\mathcal{F}(\lambda_{\rm H},\lambda_{\rm G})$ can not be obtained, the characteristics of the critical curve can be verified in simulations.
\begin{figure}[ht]
    \centering
    \includegraphics[width=1\linewidth]{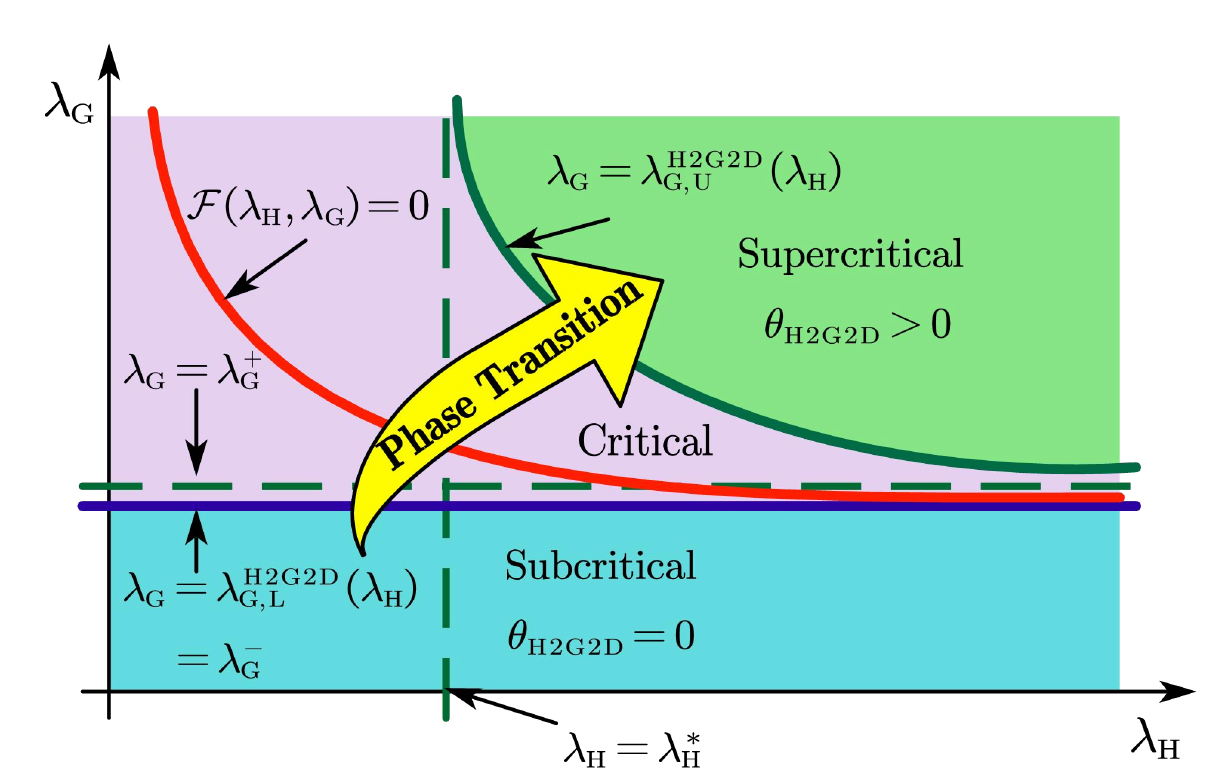}
    \caption{Sub-critical, super-critical, and critical regions for the H2G2D coverage scheme.}
    \label{fig:criticalregionH2G2D}
\end{figure}
\subsection{Hybrid Coverage}
In the hybrid coverage scheme, we discuss the sub-critical case first where the percolation probability is zero.\\

\begin{figure}[ht]
    \centering
    \includegraphics[width=0.7\linewidth]{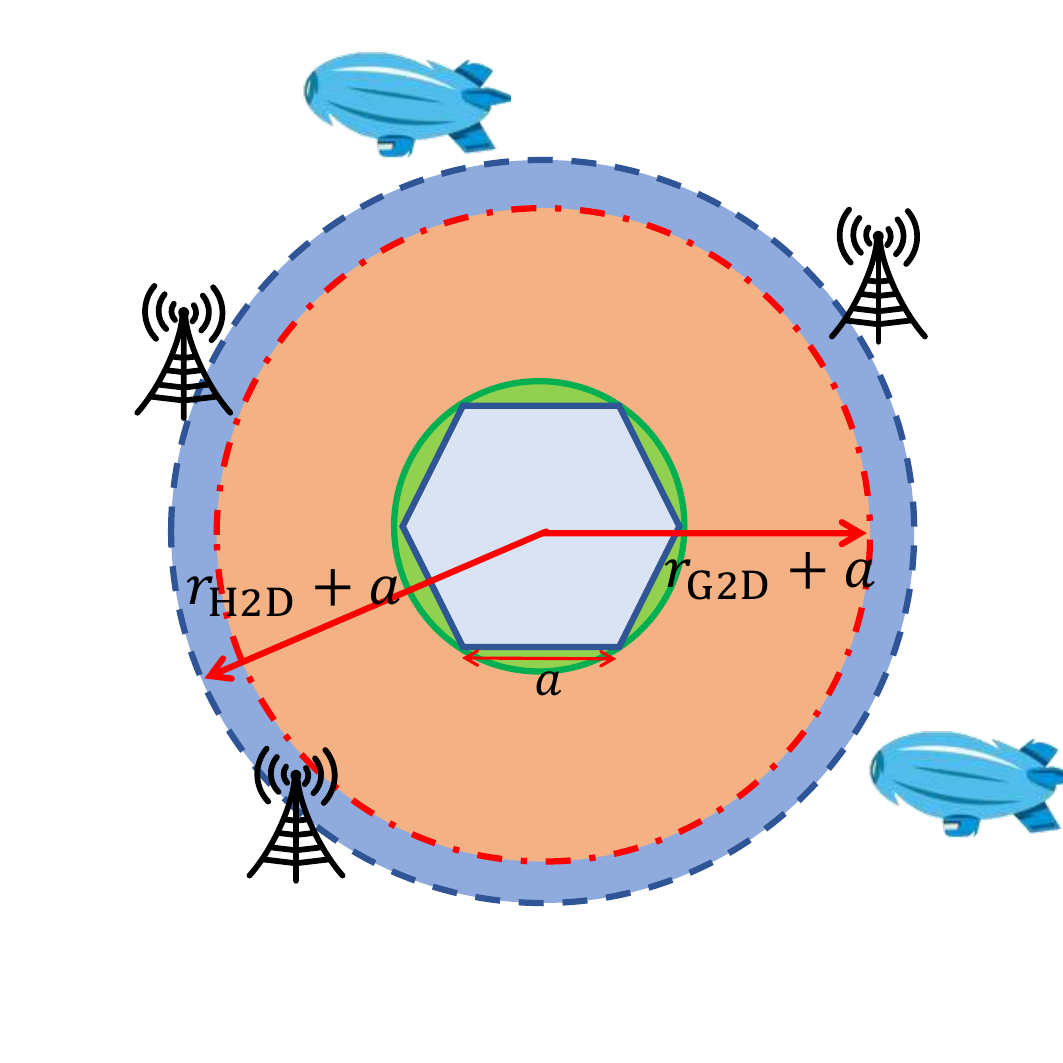}
    \caption{A closed hexagonal face $\mathcal{H}_l$ in the hybrid coverage scheme. The HAPSs and GWs are distributed outside their designed areas, respectively. Therefore, the MCC $\mathcal{T}_l$ can not be covered through H2D coverage or H2G2D coverage so that the face $\mathcal{H}_l$ is closed.}
    \label{fig:HybridClosed}
\end{figure}
\textbf{Sub-critical Case:} In the hybrid coverage scheme, we can jointly consider the H2D coverage and H2G2D coverage. Conside the MCC $\mathcal{T}_l$ of a hexagon $\mathcal{H}_l$, if there is no GW inside the circular area with radius $r_{\rm G2D}+a$ and there is no HAPS inside the circular area with radius $r_{\rm H2D}+a$, any device inside $\mathcal{T}_l$ can not be covered. Therefore, we can obtain a lower bound of $\P\{\mathcal{T}_l\;{\rm is\;not\;covered}\}$:
\begin{equation}
    \P\{\mathcal{T}_l\;{\rm is\;not\;covered}\}\geq e^{-\lambda_{\rm G}(r_{\rm G2D}+a)^2} e^{-\lambda_{\rm H}(r_{\rm H2D}+a)^2}.
\label{Tlclosedhybrid}
\end{equation}
Based on this, we introduce the sufficient condition for zero percolation probability in Theorem \ref{theo:hybridlower}.
\begin{theorem}
For a fixed value of HAPS density $\lambda_{\rm H}<\frac{\ln 2}{\pi (r_{\rm H2D}+a)^2}$, when the GW density $\lambda_{\rm G}$ and the HAPS density $\lambda_{\rm H}$ satisfies $\lambda_{\rm G}<\lambda_{\rm G,L}^{\rm hbd}(\lambda_{\rm H})$, where:
\begin{equation}
    \lambda_{\rm G,L}^{\rm hbd}(\lambda_{\rm H})=\frac{\ln 2}{\pi (r_{\rm G2D}+a)^2}-\lambda_{\rm H}\frac{(r_{\rm H2D}+a)^2}{(r_{\rm G2D}+a)^2},
\end{equation}
the percolation probability of the hybrid coverage scheme is zero, \ie $\theta_{\rm hbd}(\lambda_{\rm H},\lambda_{\rm G})=0$.

For a fixed value of GW density $\lambda_{\rm G}<\lambda_{\rm G}^{-}$, the sufficient condition for $\theta_{\rm hbd}(\lambda_{\rm H},\lambda_{\rm G})=0$ can be rewritten as:
\begin{equation}
    \lambda_{\rm H}<{\lambda_{\rm G,L}^{\rm hbd}}^{-1}(\lambda_{\rm G}),
\end{equation}
where ${\lambda_{\rm G,L}^{\rm hbd}}^{-1}(\lambda_{\rm G})$ is the inverse function of $\lambda_{\rm G,L}^{\rm hbd}(\lambda_{\rm H})$.
\label{theo:hybridlower}
\end{theorem}
\begin{IEEEproof}
    Substituting the lower bound in (\ref{Tlclosedhybrid}) into (\ref{Tlclosed}), we can obtain the sufficient condition for zero percolation probability in the hybrid coverage scheme.
\end{IEEEproof}
Next, we introduce the super-critical case where the percolation probability is non-zero.\\

\begin{figure}[ht]
    \centering
    \includegraphics[width=0.7\linewidth]{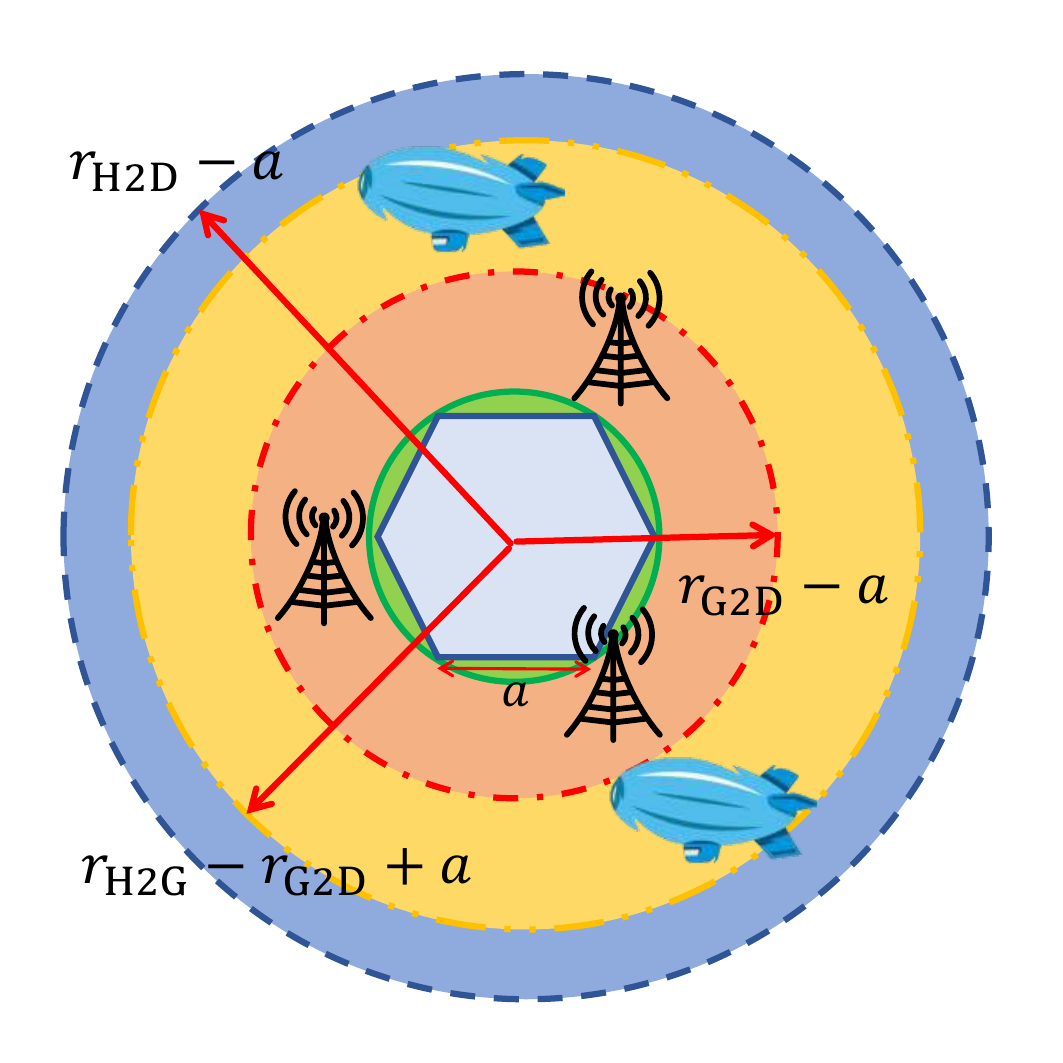}
    \caption{An open hexagonal face $\mathcal{H}_l$ in hybrid coverage scheme when $r_{\rm H2D}-a>r_{\rm H2D}-r_{\rm G2D}+a$. The HAPSs and GWs are distributed outside their designed areas, respectively. Therefore, the MCC $\mathcal{T}_l$ can be covered through H2D coverage or H2G2D coverage so that the face $\mathcal{H}_l$ is open. Especially, in this case, when HAPSs cover the hexagon via GWs, they can already cover the hexagon through H2D links.}
    \label{fig:HybridOpen1}
\end{figure}
\begin{figure}[ht]
    \centering
    \includegraphics[width=0.7\linewidth]{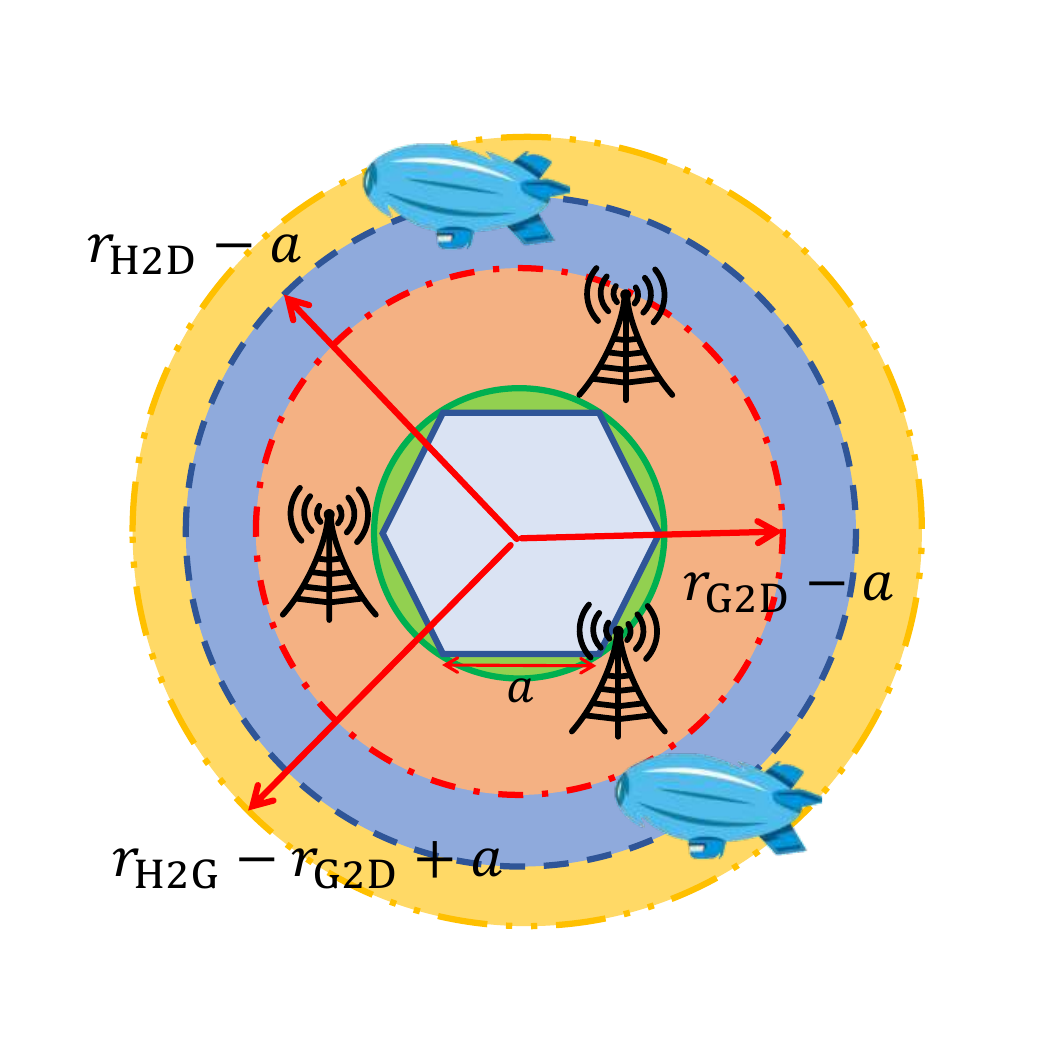}
    \caption{An open hexagonal face $\mathcal{H}_l$ in the hybrid coverage scheme when $r_{\rm H2D}-a<r_{\rm H2D}-r_{\rm G2D}+a$. The HAPSs and GWs are distributed outside their designed areas, respectively. Therefore, the MCC $\mathcal{T}_l$ can be covered through H2D coverage or H2G2D coverage so that the face $\mathcal{H}_l$ is open. }
    \label{fig:HybridOpen2}
\end{figure}

\textbf{Super-critical Case:} Consider the MCC $\mathcal{T}_l$ of the hexagon $\mathcal{H}_l$. There are two possible scenarios that ensure $\mathcal{T}_{l}$ is covered: \textit{(i) there is at least one GW inside the circular area with radius $r_{\rm G2D}-a$ and at least one HAPS inside the circular area with radius $r_{\rm H2G}-r_{\rm H2D}+a$}, or \textit{(ii) there is at least one HAPS inside the circular area with radius $r_{\rm H2D}-a$}. If $r_{\rm H2G}-r_{\rm H2D}+a>r_{\rm H2D}-a$, we can obtain:
\begin{equation}
\begin{array}{r@{}l}
    \P\{\mathcal{T}_l{\rm \;is\;covered}\}\geq e^{-\lambda_{\rm G}\pi (r_{\rm G2D}-a)^2}(1-e^{-\lambda_{\rm H}\pi (r_{\rm H2D}-a)^2})\\
+(1-e^{-\lambda_{\rm G}\pi (r_{\rm G2D}-a)^2})(1-e^{-\lambda_{\rm H}\pi (r_{\rm H2G}-r_{\rm G2D}+a)^2}).
\end{array}
\label{TlclosedHBD1}
\end{equation}
If $r_{\rm H2G}-r_{\rm H2D}+a\leq r_{\rm H2D}-a$, we can obtain:
\begin{equation}
    \P\{\mathcal{T}_l{\rm \;is\;covered}\}\geq 1-e^{-\lambda_{\rm H}\pi (r_{\rm H2D}-a)^2}.
\label{TlclosedHBD2}
\end{equation}
Based on this, we can obtain the sufficient condition for non-zero percolation probability in Theorem \ref{theo:hybridupper}.
\begin{theorem}
    When $r_{\rm H2G}-r_{\rm H2D}+a\leq r_{\rm H2D}-a$, a sufficient condition for non-zero percolation probability, \ie $\theta_{\rm hbd}(\lambda_{\rm H},\lambda_{\rm G})>0$, is 
\begin{equation}
    \lambda_{\rm H}>\lambda_{\rm H}^{+}.
\end{equation}
When $r_{\rm H2G}-r_{\rm H2D}+a>r_{\rm H2D}-a$, a sufficient condition for $\theta_{\rm hbd}(\lambda_{\rm H},\lambda_{\rm G})>0$ is 
\begin{equation}
    \lambda_{\rm H}>\lambda_{\rm H}^{+}.
\end{equation}
For a fixed value of HAPS density $\lambda_{\rm H}^{*}<\lambda_{\rm H}<\lambda_{\rm H}^{+}$, when the GW density $\lambda_{\rm G}$ satisfies $\lambda_{\rm G}>\lambda_{\rm G,U}^{\rm hbd}(\lambda_{\rm H})$  where
\begin{equation}
    \lambda_{\rm G,U}^{\rm hbd}(\lambda_{\rm H})=\frac{\displaystyle\ln (1+\frac{e^{-\lambda_{\rm H}\pi(r_{\rm H2D}-a)^2}-\frac{1}{2}}{\frac{1}{2}-e^{-\lambda_{\rm H}\pi (r_{\rm H2G}-r_{\rm G2D}+a)^2}})}{\pi (r_{\rm G2D}-a)^2}, 
\end{equation}
percolation probability of the hybrid coverage scheme is non-zero, \ie $\theta_{\rm hbd}(\lambda_{\rm H},\lambda_{\rm G})>0$.

For a fixed value of GW density $\lambda_{\rm G}$, the sufficient condition for $\theta_{\rm hbd}(\lambda_{\rm H},\lambda_{\rm G})>0$ can be rewritten as:
\begin{equation}
    \lambda_{\rm H}>{\lambda_{\rm G,U}^{\rm hbd}}^{-1}(\lambda_{\rm G}),
\end{equation}
where ${\lambda_{\rm G,U}^{\rm hbd}}^{-1}(\lambda_{\rm G})$ is the inverse function of $\lambda_{\rm G,U}^{\rm hbd}(\lambda_{\rm H})$.
\label{theo:hybridupper}
\end{theorem}
\begin{IEEEproof}
    Substituting the lower bounds in (\ref{TlclosedHBD1}) and (\ref{TlclosedHBD2}) into (\ref{Tlopen}) respectively, we can obtain the sufficient conditions for non-zero percolation probability in the hybrid coverage scheme.
\end{IEEEproof}
When $r_{\rm H2G}$ approaches infinity, all GWs can be covered by HAPS, where the upper bound for non-zero percolation probability in the hybrid coverage scheme can be simplified.
\begin{remark}
    When $r_{\rm H2G}\rightarrow \infty$, the sufficient condition approaches $\lambda_{\rm H}>\lambda_{\rm H}^{+}$ or $\lambda_{\rm G}>\lambda_{\rm G,*}^{\rm hbd}(\lambda_{\rm H})$, where:
\begin{equation}
    \lambda_{\rm G,*}^{\rm hbd}(\lambda_{\rm H})=\frac{\ln 2}{\pi (r_{\rm G2D}-a)^2}-\lambda_{\rm H}\frac{(r_{\rm H2D}-a)^2}{(r_{\rm G2D}-a)^2}
\end{equation}
and $0<\lambda_{\rm H}<\lambda_{\rm H}^{+}$. This is also the upper bound for two Gilbert disk models of GWs and HAPS with different coverage radii for wireless devices because all GWs are already covered by HAPS directly. This is also the lower bound of $\lambda_{\rm G,U}^{\rm hbd}$ and also the asymptote when $\lambda_{\rm H}$ approaches $\lambda_{\rm H}^{*}$.
    
\end{remark}

To prove the existence of the critical region for phase transition of percolation probability, we introduce the relationship between HAPS density, GW density, and percolation probability.
\begin{lemma}
For a fixed value of HAPS density $\lambda_{\rm H,1}$, the percolation probability $\theta_{\rm hbd}$ does not decrease when the GW density increases, that is
\begin{equation}
    \theta_{\rm hbd}(\lambda_{\rm H,1},\lambda_{\rm G,2})\geq \theta_{\rm hbd}(\lambda_{\rm H,1},\lambda_{\rm G,1})\;{\rm if}\;\lambda_{\rm G,2}>\lambda_{\rm G,1}.
\end{equation}
For a fixed value of GW density $\lambda_{\rm G,1}$, the percolation probability $\theta_{\rm hbd}$ does not decrease when the HAPS density increases, that is
\begin{equation}
    \theta_{\rm hbd}(\lambda_{\rm H,2},\lambda_{\rm G,1})\geq \theta_{\rm hbd}(\lambda_{\rm H,1},\lambda_{\rm G,1})\;{\rm if}\;\lambda_{\rm H,2}>\lambda_{\rm H,1}.
\end{equation}

\label{lem:increaseHybrid}
\end{lemma}
\begin{IEEEproof}
    See Appendix~\ref{app:increaseHybrid}.
\end{IEEEproof}
Therefore, with the derived sub-critical and super-critical regions, we can prove the existence of the critical condition of HAPS density and GW density for phase transition of percolation probability, which is shown in Lemma \ref{lem:criticalHybrid}. 
\begin{lemma}
    For the hybrid coverage scheme, as shown in Fig. \ref{fig:CRHybrid}, the critical condition for phase transition of percolation probability can be described as an implicit function $\mathcal{G}(\lambda_{\rm H},\lambda_{\rm G})=0$, whose curve is between the super-critical region and sub-critical region. It has five properties: (i) the increase in $\lambda_{\rm H}$ does not lead to a higher critical value of $\lambda_{\rm G}$, (ii) the increase in $\lambda_{\rm G}$ does not lead to a higher critical value of $\lambda_{\rm H}$, (iii) the phase transition does not exist in the region where $\lambda_{\rm H}>\lambda_{\rm H}^{+}$, (iv) when $r_{\rm H2G}-r_{\rm H2D}+a>r_{\rm H2D}-a$, for $\lambda_{\rm H}^{*}<\lambda_{\rm H}<\lambda_{\rm H}^{+}$, the critical value of $\lambda_{\rm G}$ is between $\max\{0,\lambda_{\rm G,L}^{\rm hbd}(\lambda_{\rm H})\}$ and $\lambda_{\rm G,U}^{\rm hbd}(\lambda_{\rm H})$, (v) for $\lambda_{\rm G}>0$, the critical value of $\lambda_{\rm H}$ is between  $\max\{0,{\lambda_{\rm G,L}^{\rm hbd}}^{-1}(\lambda_{\rm G})\}$ and ${\lambda_{\rm G,U}^{\rm hbd}}^{-1}(\lambda_{\rm G})$.
\label{lem:criticalHybrid}
\end{lemma}
\begin{IEEEproof}
    See Appendix~\ref{app:criticalhybrid}.
\end{IEEEproof}

By observing the sub-critical and super-critical regions of the H2D, H2G2D, and hybrid coverage schemes, we can find out the relationship between them.

\begin{remark}
    Because the random graphs satisfy $G_{\rm H2D}\subseteq G_{\rm hbd}$ and $G_{\rm H2G2D}\subseteq G_{\rm hbd}$, we can also know that $\theta_{\rm hbd}(\lambda_{\rm H},\lambda_{\rm G})\geq \theta_{\rm H2D}(\lambda_{\rm H},\lambda_{\rm G})$ and $\theta_{\rm hbd}(\lambda_{\rm H},\lambda_{\rm G})\geq \theta_{\rm H2G2D}(\lambda_{\rm H},\lambda_{\rm G})$. The sub-critical region of the hybrid coverage scheme is contained by the H2D sub-critical region or the H2G2D sub-critical region. Similarly, the super-critical region of the hybrid coverage scheme contains both of the H2D super-critical region and the H2G2D super-critical region.
\end{remark}

\begin{figure}[htbp]
\centering
\subfigure[Case 1: Critical region in hybrid coverage scheme when $r_{\rm H2G}-r_{\rm G2D}+a\leq r_{\rm H2D}-a\leq r_{\rm H2D}+a$.]{
\begin{minipage}[t]{1\linewidth}
\centering 
\includegraphics[width=1\textwidth]{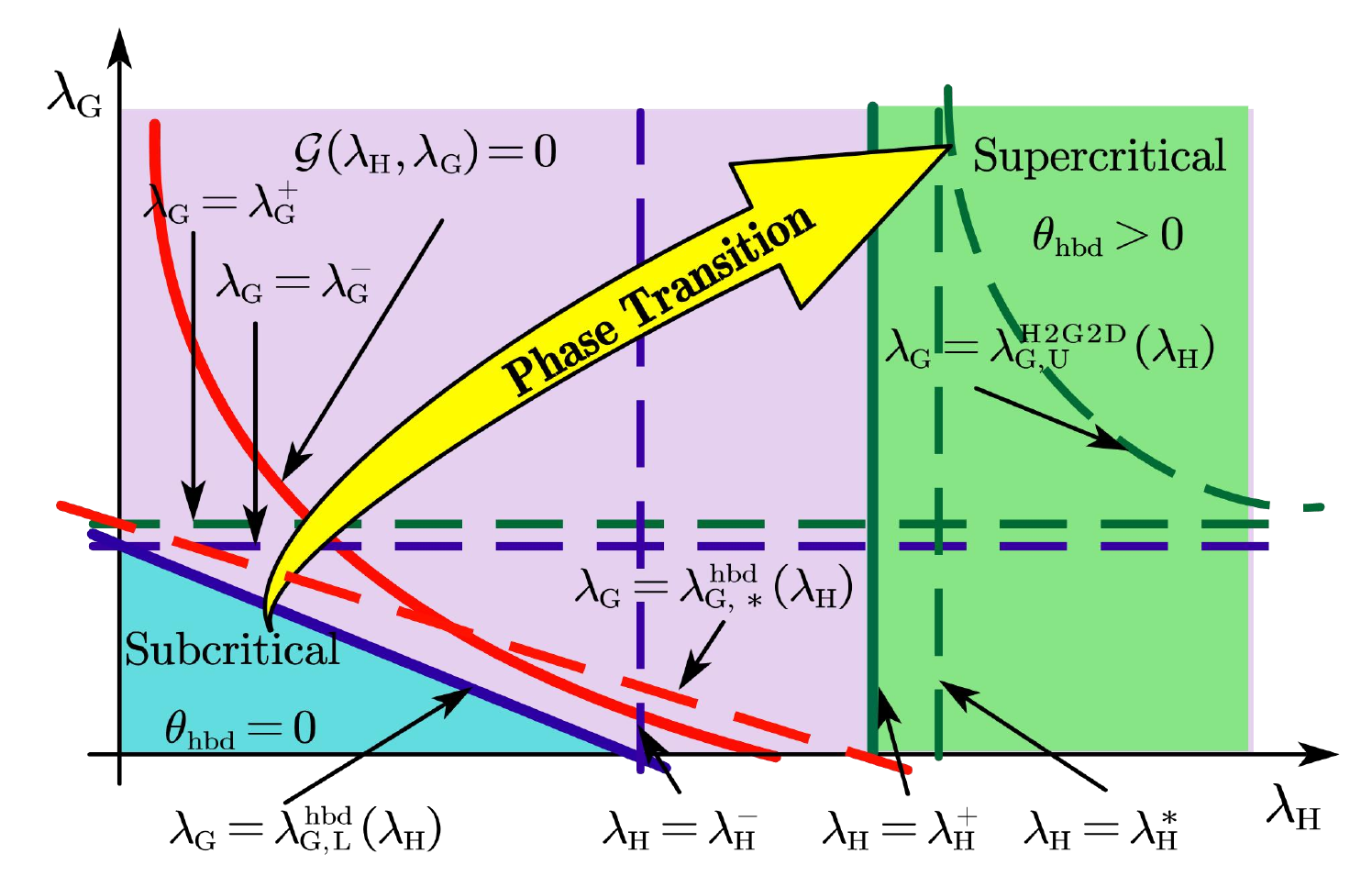}
\label{fig:CRHybrid1}
\end{minipage}}

\subfigure[Case 2: Critical region in hybrid coverage scheme when $ r_{\rm H2D}-a\leq r_{\rm H2G}-r_{\rm G2D}+a\leq r_{\rm H2D}+a$.]{
\begin{minipage}[t]{1\linewidth}
\centering 
\includegraphics[width=1\textwidth]{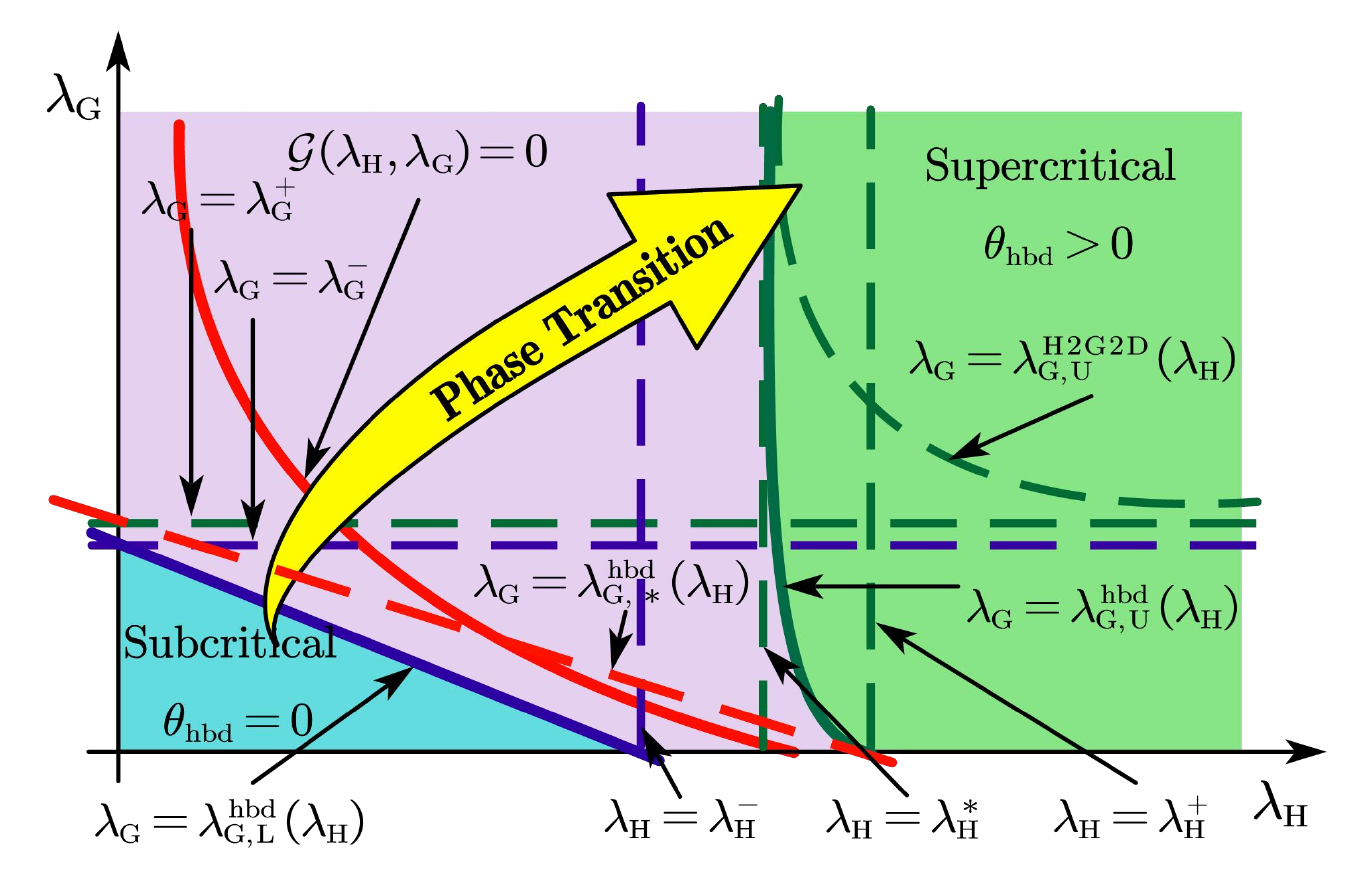}
\label{fig:CRHybrid2}
\end{minipage}}

\subfigure[Case 3: Critical region in hybrid coverage scheme when $ r_{\rm H2D}-a\leq r_{\rm H2D}+a\leq r_{\rm H2G}-r_{\rm G2D}+a$.]{
\begin{minipage}[t]{1\linewidth}
\centering 
\includegraphics[width=1\textwidth]{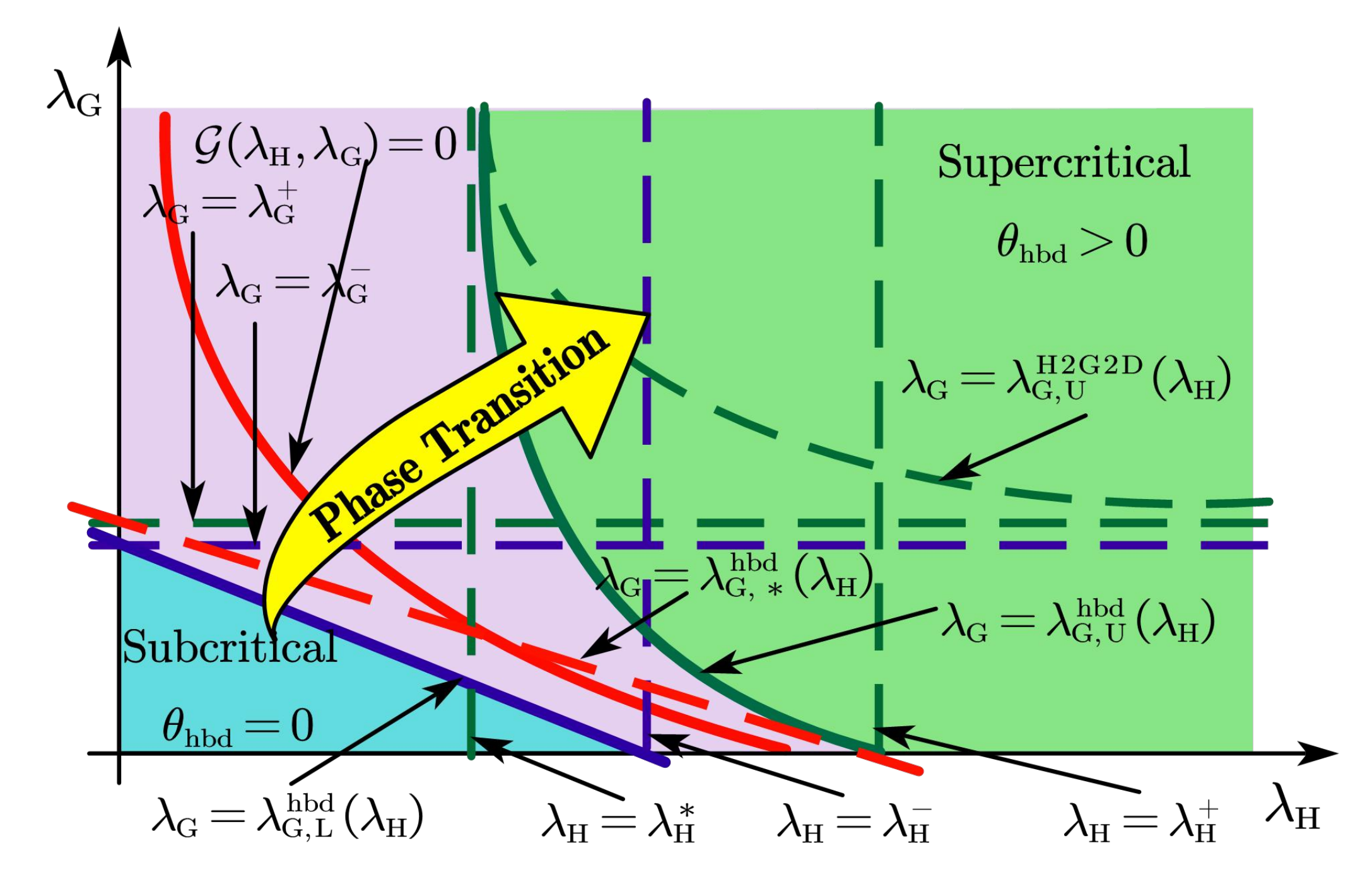}
\label{fig:CRHybrid3}
\end{minipage}}
\caption{Sub-critical, super-critical, and critical region for the hybrid coverage scheme.}
\label{fig:CRHybrid}
\end{figure}

For all of these three HAPS-based solutions, we summarize the necessary expressions for ease of reading and comparison in Table \ref{tab:bounds}.

\begin{table*}[ht]
\caption{Table of Bounds}
\label{tab:bounds}
\centering
\begin{center}
\begin{tabular}{|llllll|}
\hline
\multicolumn{2}{|l|}{}       & \multicolumn{2}{l|}{Lower Bound} & \multicolumn{2}{l|}{Upper bound} \\ \hline
\multicolumn{2}{|l|}{H2D}    & \multicolumn{2}{l|}{$
    \lambda_{\rm H,L}^{\rm H2D}=\lambda_{\rm H}^{-}\overset{\triangle}{=}\frac{\ln 2}{\pi (r_{\rm H2D}+a)^2}$}            & \multicolumn{2}{l|}{$
    \lambda_{\rm H,U}^{\rm H2D}=\lambda_{\rm H}^{+}\overset{\triangle}{=}\frac{\ln 2}{\pi (r_{\rm H2D}-a)^2}$}            \\ \hline
\multicolumn{2}{|l|}{H2G2D}  & \multicolumn{2}{l|}{$\lambda_{\rm G,L}^{\rm H2G2D}(\lambda_{\rm H})=\lambda_{\rm G}^{-}\overset{\triangle}{=}\frac{\ln 2}{\pi(r_{\rm G2D}+a)^2}$}            & \multicolumn{2}{l|}{$\lambda_{\rm G,U}^{\rm H2G2D}(\lambda_{\rm H})=\frac{\ln (1+\frac{\frac{1}{2}}{\frac{1}{2}-e^{-\lambda_{\rm H}\pi(r_{\rm H2G}-r_{\rm G2W}+a)^2}})}{\pi(r_{\rm G2D}-a)^2}$ for $\lambda_{\rm H}>\lambda_{\rm H}^{*}$}            \\ \hline
\multicolumn{2}{|l|}{Hybrid} & \multicolumn{2}{l|}{$\lambda_{\rm G,L}^{\rm hbd}(\lambda_{\rm H})=\frac{\ln 2}{\pi (r_{\rm G2D}+a)^2}-\lambda_{\rm H}\frac{(r_{\rm H2D}+a)^2}{(r_{\rm G2D}+a)^2}$}            & \multicolumn{2}{l|}{$\lambda_{\rm H,U}^{\rm hbd}=\lambda_{\rm H}^{+}$ or $\lambda_{\rm G,U}^{\rm hbd}(\lambda_{\rm H})=\frac{\ln (1+\frac{e^{-\lambda_{\rm H}\pi(r_{\rm H2D}-a)^2}-\frac{1}{2}}{\frac{1}{2}-e^{-\lambda_{\rm H}\pi (r_{\rm H2G}-r_{\rm G2D}+a)^2}})}{\pi (r_{\rm G2D}-a)^2}$ for $\lambda_{\rm H}^{*}<\lambda_{\rm H}<\lambda_{\rm H}^{+}$}            \\ \hline
\multicolumn{6}{|c|}{$\lim\limits_{r_{\rm H2G}\rightarrow\infty}\lambda_{\rm G,L}^{\rm H2G2D}(\lambda_{\rm H})=\lambda_{\rm G}^{+}\overset{\triangle}{=}\frac{\ln 2}{\pi(r_{\rm G2D}-a)^2}$,\;$\lambda_{\rm H}^{*}\overset{\triangle}{=}\frac{\ln 2}{\pi (r_{\rm H2G}-r_{\rm G2D}+a)^2}$}                            \\ \hline

\multicolumn{6}{|c|}{$\lim\limits_{r_{\rm H2G}\rightarrow\infty}\lambda_{\rm G,L}^{\rm hbd}(\lambda_{\rm H})=\lambda_{\rm G,*}^{\rm hbd}(\lambda_{\rm H})=\frac{\ln 2}{\pi (r_{\rm G2D}-a)^2}-\lambda_{\rm H}\frac{(r_{\rm H2D}-a)^2}{(r_{\rm G2D}-a)^2}$}                            \\ \hline
\end{tabular}
\end{center}
\end{table*}


\section{Simulation results and discussion}\label{sec:simulation}

\indent In this paper, we proved the existence of critical conditions for phase transition of percolation probability from zero to non-zero in the H2D coverage scheme, H2G2D coverage scheme, and hybrid coverage scheme. We take the HAPS-enabled IoT network as an example,  assuming that HAPSs operate at the altitude $h_{\rm H}=20\;{\rm km}$ and the altitudes of IoT GWs $h_{\rm G}=10\;{\rm m}$. The potential wireless devices are deployed on the ground so that $h_{\rm D}=0\;{\rm m}$. We conduct the simulation in a $400\;{\rm km}\times 400\;{\rm km}$ area. The H2D, H2G, G2D, G2G communication ranges are $d_{\rm H2D}=30\;{\rm km}$, $d_{\rm H2G}=50\;{\rm km}$, $d_{\rm G2D}=10\;{\rm km}$ and $r_{\rm G2G}=15\;{\rm km}$ \cite{9380673,arum2020review,soy2023coverage}.

\begin{figure}[ht]
    \centering
    \includegraphics[width=1\linewidth]{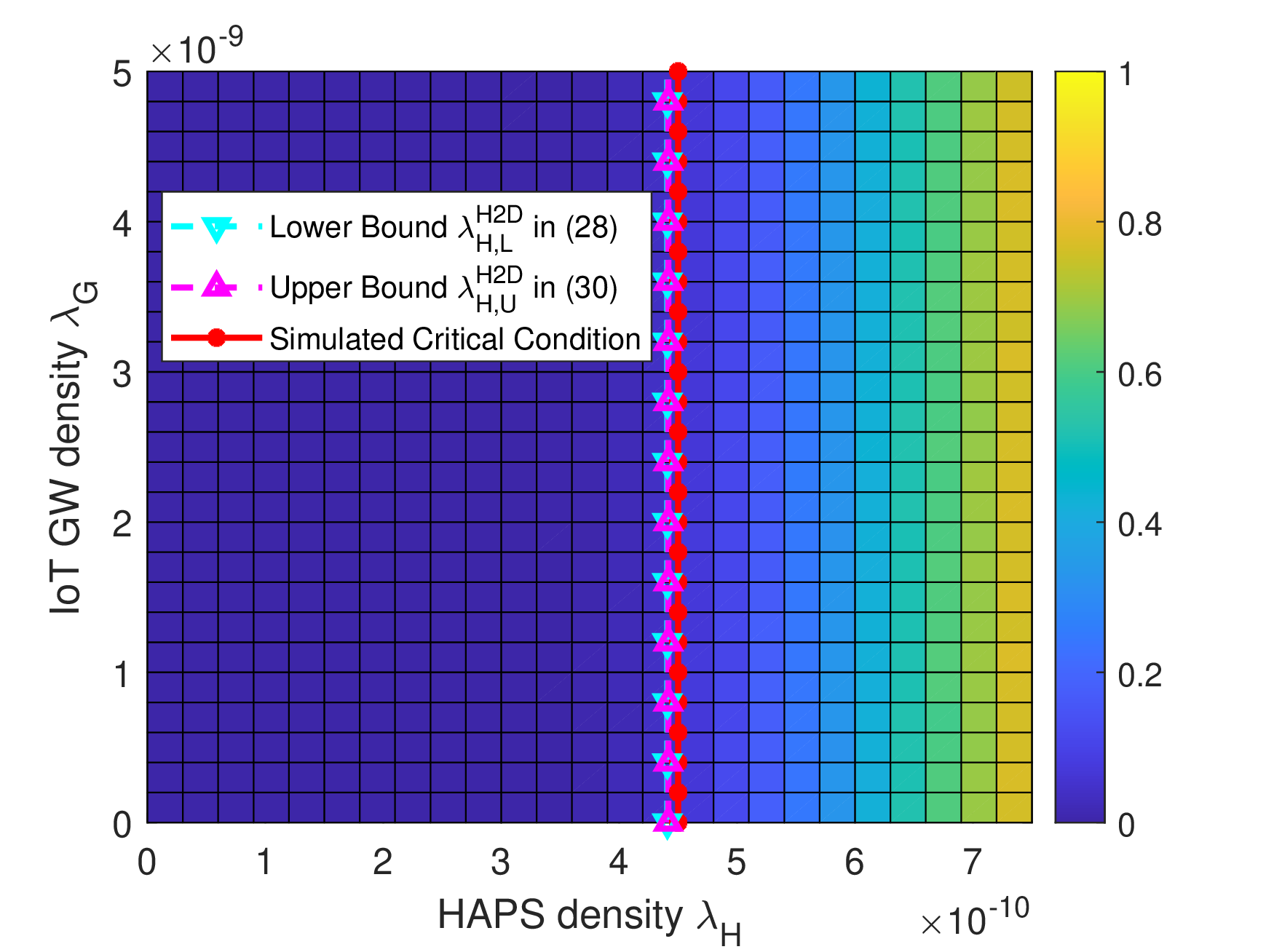}
    \caption{Heatmap of percolation probability in H2D coverage scheme with different HAPS densities and GW densities.}
    \label{fig:H2Dheatmap}
\end{figure}
\indent In Fig. \ref{fig:H2Dheatmap}, we draw the heatmap of percolation probability in H2D coverage scheme, considering different combinations of HAPS density $\lambda_{\rm H}$ and GW density $\lambda_{\rm G}$.  The phase transition from the sub-critical case to the super-critical case relies on the increase in HAPS density $\lambda_{\rm H}$, which matches the concept that we show in Lemma \ref{lem:criticalH2D}. When the HAPS density exceeds $4.5\times 10^{-10}\;{\rm HAPSs/km^2}$, the percolation probability rapidly increases from $0.1$, and reach $0.83$ when $\lambda_{\rm H}=7.5\times 10^{-10}\;{\rm HAPSs/km^2}$. In this coverage scheme, percolation probability only depends on the HAPS density when the footprint or associating area of each HAPS is fixed.\\ 
\begin{figure}[ht]
    \centering
    \includegraphics[width=1\linewidth]{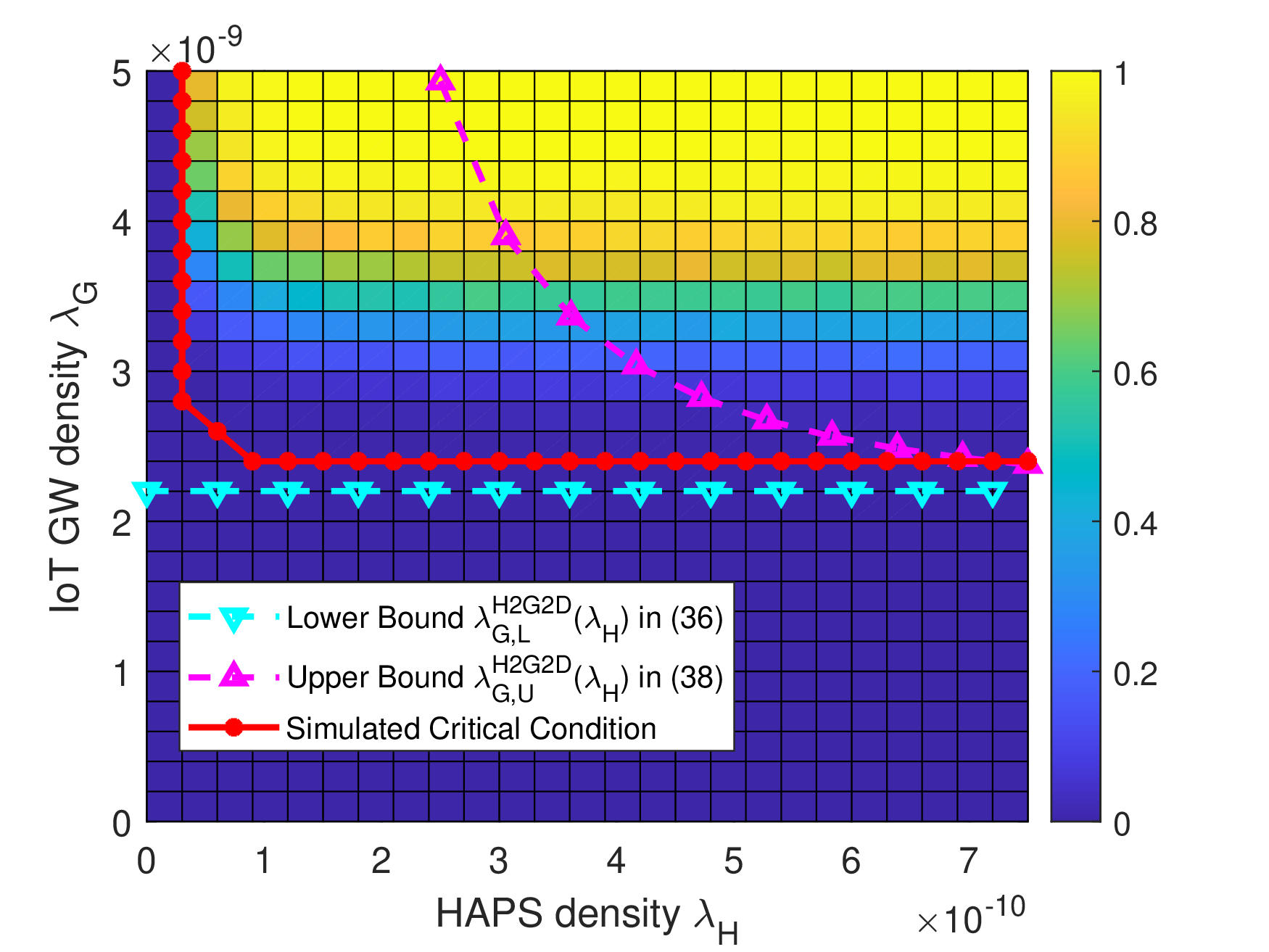}
    \caption{Heatmap of percolation probability in H2G2D coverage scheme with different HAPS densities and GW densities.}
    \label{fig:H2G2Dheatmap}
\end{figure}
\indent In Fig. \ref{fig:H2G2Dheatmap}, we draw the heatmap of percolation probability in H2G2D coverage scheme with different combinations of $\lambda_{\rm H}$ and $\lambda_{\rm G}$. The phase transition from the sub-critical case to the super-critical case relies on the increase in not only $\lambda_{\rm H}$ but also $\lambda_{\rm G}$. Following Lemma \ref{lem:criticalH2G2D}, higher HAPS density $\lambda_{\rm H}$ leads to lower requirement of GW density $\lambda_{\rm G}$. When $\lambda_{\rm H}$ is higher than $9\times 10^{-11}\;{\rm HAPSs/km^2}$, the requirement of GW density approaches to a constant $2\times 10^{-9}\;{\rm GWs/km^2}$. The set of points on critical state are located between the bounds of sub-critical region and super-critical region. When $\lambda_{\rm G}$ increases, the percolation probability rapidly increases from $0.01$. When $\lambda_{\rm G}=3\times 10^{-9}\;{\rm GWs/km^2}$, the increase in HAPS density can not obtain a rapid performance improvement because there is no enough GWs to form the continuous service area. However, when $\lambda_{\rm G}=4\times 10^{-9}\;{\rm GWs/km^2}$, the increase in HAPS density can make the percolation probability be promoted rapidly. In short, when HAPS density is sufficiently high, because the devices are connected to the GWs directly, we only need to deploy enough GWs to ensure that most neighbor GWs have overlapping coverage. When HAPS density is insufficient, the communications between GWs can help eliminate the negative impact of insufficient HAPS deployment, and sufficient density is also important to activate the GWs through G2G links. Therefore, in the H2G2D coverage scheme, the first task is to make sure that the GW mesh network is well-deployed. \\
\begin{figure}[ht]
    \centering
    \includegraphics[width=1\linewidth]{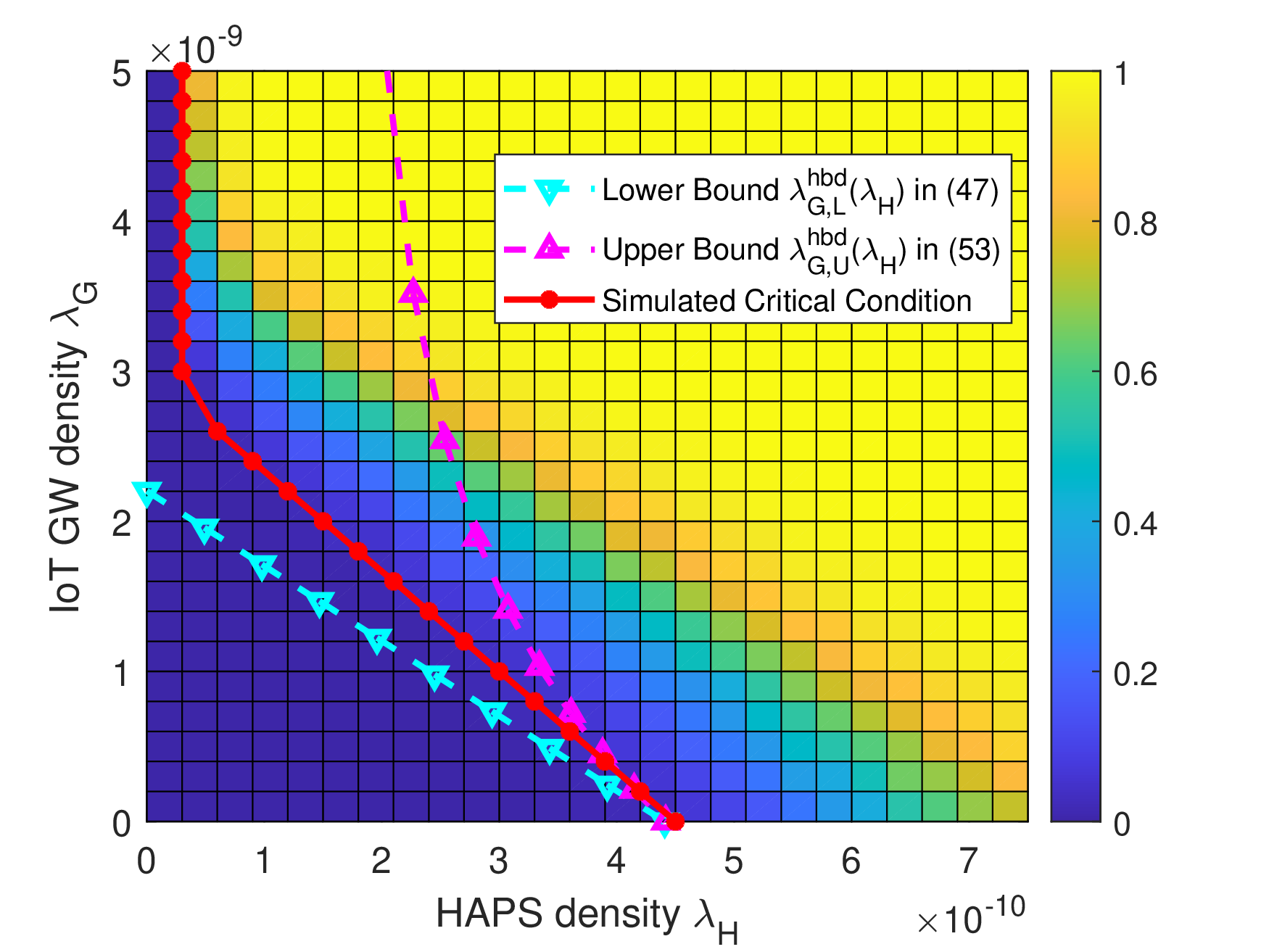}
    \caption{Heatmap of percolation probability in hybrid coverage scheme with different HAPS densities and GW densities.}
    \label{fig:Hybridheatmap}
\end{figure}
\indent In Fig. \ref{fig:Hybridheatmap}, we draw the heatmap of percolation probability in the hybrid coverage scheme with different combinations of $\lambda_{\rm H}$ and $\lambda_{\rm G}$. Similar to the H2G2D case, the phase transition from the sub-critical case to the super-critical case relies on the increase in $\lambda_{\rm H}$ and also $\lambda_{\rm G}$. Differently, the sub-critical region is smaller than those in H2D scheme and H2G2D scheme, while the super-critical region is larger or equal to those in H2D scheme and H2G2D scheme. The set of points on the critical state are located between the bounds of sub-critical region and super-critical region. When $\lambda_{\rm H}>4.5\times 10^{-10}\;{\rm GWs/km^2}$, the percolation probability is non-zero. When $\lambda_{\rm H}<4.5\times 10^{-10}\;{\rm GWs/km^2}$, percolation probability increases rapidly from $0.1$ when the HAPS density is larger than the corresponding critical value. For any fixed GW density, the increase in HAPS density can always bring the performance improvement, however, the increase in $\lambda_{\rm G}$ leads to a less critical value of $\lambda_{\rm H}$. For example, when $\lambda_{\rm G}=1\times10^{-9}\;{\rm GWs/km^2}$, the critical value of $\lambda_{\rm H}$ is $3\times10^{-10}\;{\rm HAPSs/km^2}$. But when $\lambda_{\rm G}=2\times10^{-9}\;{\rm GWs/km^2}$, the critical value of $\lambda_{\rm H}$ is $1.5\times10^{-10}\;{\rm HAPSs/km^2}$. Therefore, in the hybrid coverage scheme, a high HAPS density only requires GWs to fill the gaps between HAPS coverage using few GW deployment, but a low HAPS density requires more GWs to generate large-scale mesh networks and continuous serving areas.\\
\begin{figure}[!ht]
    \centering
    \includegraphics[width=0.9\linewidth]{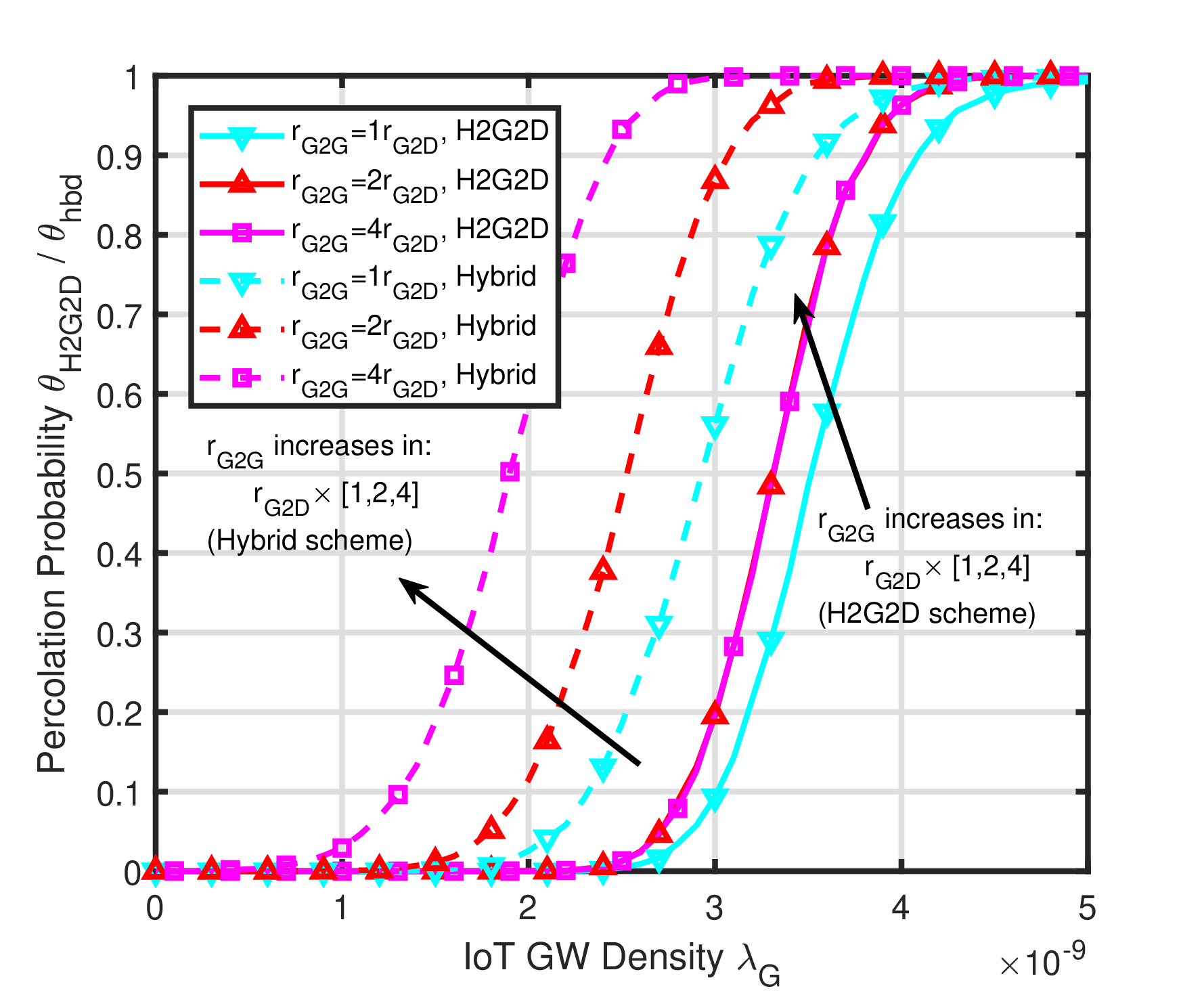}
    \caption{Percolation probability vs GW density with different value of $r_{\rm G2G}$ and a fixed HAPS density in H2G2D coverage scheme and hybrid coverage scheme.}
    \label{fig:rG2G}
\end{figure}
\indent We notice that the hybrid coverage scheme can obtain a higher percolation probability by jointing the H2D and H2G2D coverage schemes. However, when the antenna array, spectrum resource and energy are limited, the performance of the hybrid coverage scheme depends on the resource allocation strategy. When all resources are allocated to the H2D function, the hybrid coverage scheme is reduced to the H2D coverage scheme. Similarly, when all resources are allocated to the H2G2D function, the hybrid coverage scheme is reduced to the H2G2D coverage scheme. Therefore, we can optimize the resource allocation strategy to obtain better connectivity performance.\\
\indent For the H2G2D coverage scheme and the hybrid coverage scheme, the G2G communications between GWs build a mesh network, where the coverage expansion can be realized via multi-hop G2G links. In Fig. \ref{fig:rG2G}, we show that the increase in G2G communication range $r_{\rm G2G}$ can lead to the performance improvement. However, when $r_{\rm G2G}>2r_{\rm G2D}$, the increase in $r_{\rm G2G}$ can not further lead to the increase in percolation probability. Differently, for the hybrid coverage scheme, the increase in $r_{\rm G2G}$ can lead to the increase in percolation probability even
when $r_{\rm G2G}>2r_{\rm G2D}$, since it brings the opportunities for H2D coverage and H2G2D coverage being connected together. \\
\indent Therefore, from the proposed perspective of percolation theory, the operators can quickly find out the critical conditions for HAPS density and GW density, where they can use minimal number of HAPSs or GWs to realize the continuous service path they want. The operators can optimize the choice of HAPS density and GW density, based on the CAPEX and OPEX for each HAPS or each GW. Take the hybrid coverage scheme as an example. The CAPEX and OPEX of HAPSs and GWs are predictable. When the budget is fixed, we have multiple combinations of HAPS density and GW density, and we can select the one with the highest percolation probability. When the requirement of percolation probability is fixed, we can draw contour lines with percolation probability equal to a certain threshold, and select the density combination with the lowest expenditure. In the case where the GW deployment or HAPS deployment is limited (such as maritime areas with a small number of GW locations or protected areas with regulated HAPS deployment), we can further constrain the densities. However, one challenge is that there is a lack of closed-form expression for the percolation probability according to the size of whole region and characteristics of nodes, so it is necessary to use some approximate expressions or directly use experimental data to obtain the best density combination.\\
\indent In this manuscript, all HAPSs and GWs are assumed to be available. In actual situations, operators can obtain the required available node densities through this framework and then derive the required node densities by considering the node failure rates. The weather condition is also an important factor in designing the HAPS deployment. For regions with stable weather conditions, operators can use the average communication distance to approximate the required node densities. For regions with variable weather, operators can consider the communication distance under different weather conditions separately, to determine the range of required node densities. In high-traffic areas, the associating area of each HAPS will be reduced. The framework in this manuscript can help operators quickly estimate the required density of HAPSs and dynamically adjust the movement of HAPSs. By changing the height of devices $h_H$, the proposed framework can fit different use cases. For example, $h_H=0\;{\rm m}$ works for IoT devices, while $h_H=1\;{\rm m}$ matches the typical height of mobile users. For future aerial vehicles, $h_H$ is expected to be set as hundreds of meters. In addition, when the altitudes of wireless devices are different or changing, especially for mountain and post-disaster areas, 3D service continuity is expected to be investigated.

In the future, dynamic percolation analysis considering the random shadowing effect, instantaneous SINR, traffic changes and spatial correlations between nodes is expected to be developed. We realize that there exists some level of mutual correlation among different nodes (GWs, HAPS, users) that is not captured in our framework for analytical tractability. This should be an interesting research direction to extend the work in this paper. Considering the spectrum competition between HAPSs and terrestrial networks and different levels of shadowing impacts, the density of HAPSs may differ depending on existing network deployments and environmental conditions. Therefore, it is essential to investigate the coverage performance and connectivity of inhomogeneous HAPS deployments and more possible vHetNet architectures. Advanced traffic management strategies on HAPSs and GWs should be developed in these proposed HAPS-based solutions. In this paper, we discuss the feasibility of HAPS-based solutions in achieving large-scale continuous service based on hexagonal face percolation. We encourage discussions based on different lattice, numerical methods, and other percolation theory tools to study and improve research on coverage connectivity analysis of wireless networks.
\section{Conclusion}\label{sec:conclusion}
In this paper, we introduce three HAPS-based solutions for large-scale continuous service: (i) the H2D coverage scheme, (ii) the H2G2D coverage scheme and (iii) the hybrid coverage scheme. We investigate the phase transition behavior of percolation probability, which represents the probability of generating large-scale continuous coverage areas via HAPSs' direct coverage or via GWs which are connected to HAPSs. We show that the critical condition for phase transition exists between the sub-critical region and the super-critical region, in all of these three HAPS-based coverage schemes. For each coverage scheme, the critical condition can form a curve about HAPS density and GW density, which can help minimize the cost from HAPSs or GWs, or jointly minimize the CAPEX and OPEX of HAPS-based wireless networks.  
\appendices
\section{Proof of Lemma \ref{lem:increaseH2D} }\label{app:increaseH2D}
In the H2D coverage scheme, we consider two sets of HAPSs $\Psi_1$ and $\Psi_2$ with densities $\lambda_{\rm H,1}$ and $\lambda_{\rm H,2}$, respectively, where $0<\lambda_{\rm H,1}< \lambda_{\rm H,2}$. Since $\Psi_1$ and $\Psi_2$ are both PPPs, $\Psi_1$ can be constructed by thinning $\Psi_2$ with probability $\frac{\lambda_{\rm H,1}}{\lambda_{\rm H,2}}$ and $\Psi_1\subseteq\Psi_2$. $K_{\rm H2D,1}(0)$ and $K_{\rm H2D,2}(0)$ are the connected components containing the origin in the random graph $G_{\rm H2D,1}$ and $G_{\rm H2D,2}$, respectively. Because $V_{\rm H2D,1}\subseteq V_{\rm H2D,2}$ and $E_{\rm H2D,1}\subseteq E_{\rm H2D,2}$, we can obtain $K_{\rm H2D,1}(0)\subseteq K_{\rm H2D,2}(0)$ and $|K_{\rm H2D,1}(0)|\leq |K_{\rm H2D,2}(0)|$. Therefore, $0<\lambda_{\rm H,1}<\lambda_{\rm H,2}$ indicates $\theta_{\rm H2D}(\lambda_{\rm H,1})\leq\theta_{\rm H2D}(\lambda_{\rm H,2})$. That means the percolation probability is a non-decreasing function of $\lambda_{\rm H}$.

\section{Proof of Lemma \ref{lem:increaseH2G2D}}\label{app:increaseH2G2D}
For a fixed value of HAPS density $\lambda_{\rm H,1}$, we consider the same set of HAPSs' 2D locations $\Psi$ with density $\lambda_{\rm H,1}$, and two different sets of GWs' 2D locations $\Phi_1$ and $\Phi_2$ with densities $\lambda_{\rm G,1}$ and $\lambda_{\rm G,2}$, respectively, where $0<\lambda_{\rm G,1}<\lambda_{\rm G,2}$. Since $\Phi_1$ and $\Phi_2$ are both PPPs, $\Phi_1$ can be constructed by thinning $\Phi_2$ with probability $\frac{\lambda_{\rm G,1}}{\lambda_{\rm G,2}}$ and $\Phi_1\subseteq\Phi_2$. The set $\Phi_2\backslash\Phi_1$ can be considered a PPP with density $\lambda_{\rm G,2}-\lambda_{\rm G,1}$. Appending any point in $\Phi_2\backslash\Phi_1$ into $\Phi_1$ does not remove any point in $\Phi_{c,1}$ because any point in $\Phi_{c,1}$ has already generated the path to HAPSs via other points in $\Phi_{c,1}$. The appended GW can be covered by HAPSs directly and join in $\Phi_{d,2}$. It can also build a G2G link to other GWs in $\Phi_{d,1}$ or $\Phi_{id,1}$, and then join in $\Phi_{id,2}$. Therefore, we can obtain $\Phi_{d,1}\subseteq\Phi_{d,2}$ and $\Phi_{id,1}\subseteq\Phi_{id,2}$. Because $\Phi_{c,1}= \Phi_{d,1}\cup \Phi_{id,1}$ and $\Phi_{c,2}= \Phi_{d,2}\cup \Phi_{id,2}$, we obtain that $\Phi_{c,1}\subseteq\Phi_{c,2}$, $V_{\rm H2G2D,1}\subseteq V_{\rm H2G2D,2}$, and $E_{\rm H2G2D,1}\subseteq E_{\rm H2G2D,2}$. Therefore, the giant component also satisfy the relationship $K_{\rm H2G2D,1}(0)\subseteq K_{\rm H2G2D,2}(0)$, and percolation probabilities satisfy:
\begin{equation}
    \theta_{\rm H2G2D}(\lambda_{\rm H,1},\lambda_{\rm G,2})\geq \theta_{\rm H2G2D}(\lambda_{\rm H,1},\lambda_{\rm G,1})
\end{equation}
and percolation probability in the H2G2D coverage scheme is a non-decreasing function of $\lambda_{\rm G}$.\\
\indent Next, for a fixed value of GW density $\lambda_{\rm G,1}$, we consider the same set of GWs $\Phi$ with density $\lambda_{\rm G,1}$, and two different sets of HAPSs $\Psi_{1}$ and $\Psi_{2}$ with densities $\lambda_{\rm G,1}$ and $\lambda_{\rm G,2}$, respectively, where $0<\lambda_{\rm G,1}<\lambda_{\rm G,2}$. Since $\Psi_1$ and $\Psi_2$ are both PPPs, $\Psi_1$ can be constructed by thinning $\Psi_2$ with probability $\frac{\lambda_{\rm H,1}}{\lambda_{\rm H,2}}$ and $\Psi_{1}\subseteq\Psi_{2}$. The set $\Psi_{2}\backslash\Psi_{1}$ can be considered a PPP with density $\lambda_{\rm H,2}-\lambda_{\rm H,1}$. Appending any point in $\Psi_{2}\backslash\Psi_{1}$ into $\Psi_{1}$ makes the below changes. First, more GWs can be directly covered by HAPSs, that is $\Phi_{d,1}\subseteq \Phi_{d,2}$. Second, other GWs in $\Phi\backslash\Phi_{c,1}$ have the chance to join in $\Phi_{id,2}$ through the GWs in $\Phi_{d,2}\backslash\Phi_{d,1}$. Third, the GWs that connect to HAPSs via G2G communications can be directly covered by the new HAPS or still covered via previous multi-hop G2G paths, that is $\Phi_{id,1}\subseteq(\Phi_{d,2}\cup\Phi_{id,2})$.   Therefore, we can obtain $(\Phi_{d,1}\cup\Phi_{id,1})\subseteq (\Phi_{d,2}\cup\Phi_{id,2})$, that is $\Phi_{c,1}\subseteq \Phi_{c,2}$. The vertex sets and edge sets also satisfy $V_{\rm H2G2D,1}\subseteq V_{\rm H2G2D,2}$ and $E_{\rm H2G2D,1}\subseteq E_{\rm H2G2D,2}$, and the giant components satisfy $K_{\rm H2G2D,1}(0)\subseteq K_{\rm H2G2D,2}(0)$ and $|K_{\rm H2G2D,1}(0)|\leq |K_{\rm H2G2D,2}(0)|$ at the same time. Therefore, we have
\begin{equation}
    \theta_{\rm H2G2D}(\lambda_{\rm H,2},\lambda_{\rm G,1})\geq \theta_{\rm H2G2D}(\lambda_{\rm H,1},\lambda_{\rm G,1})
\end{equation}
and percolation probability in the H2G2D coverage scheme is a non-decreasing function of $\lambda_{\rm H}$.
\begin{figure}
    \centering
    \includegraphics[width=1\linewidth]{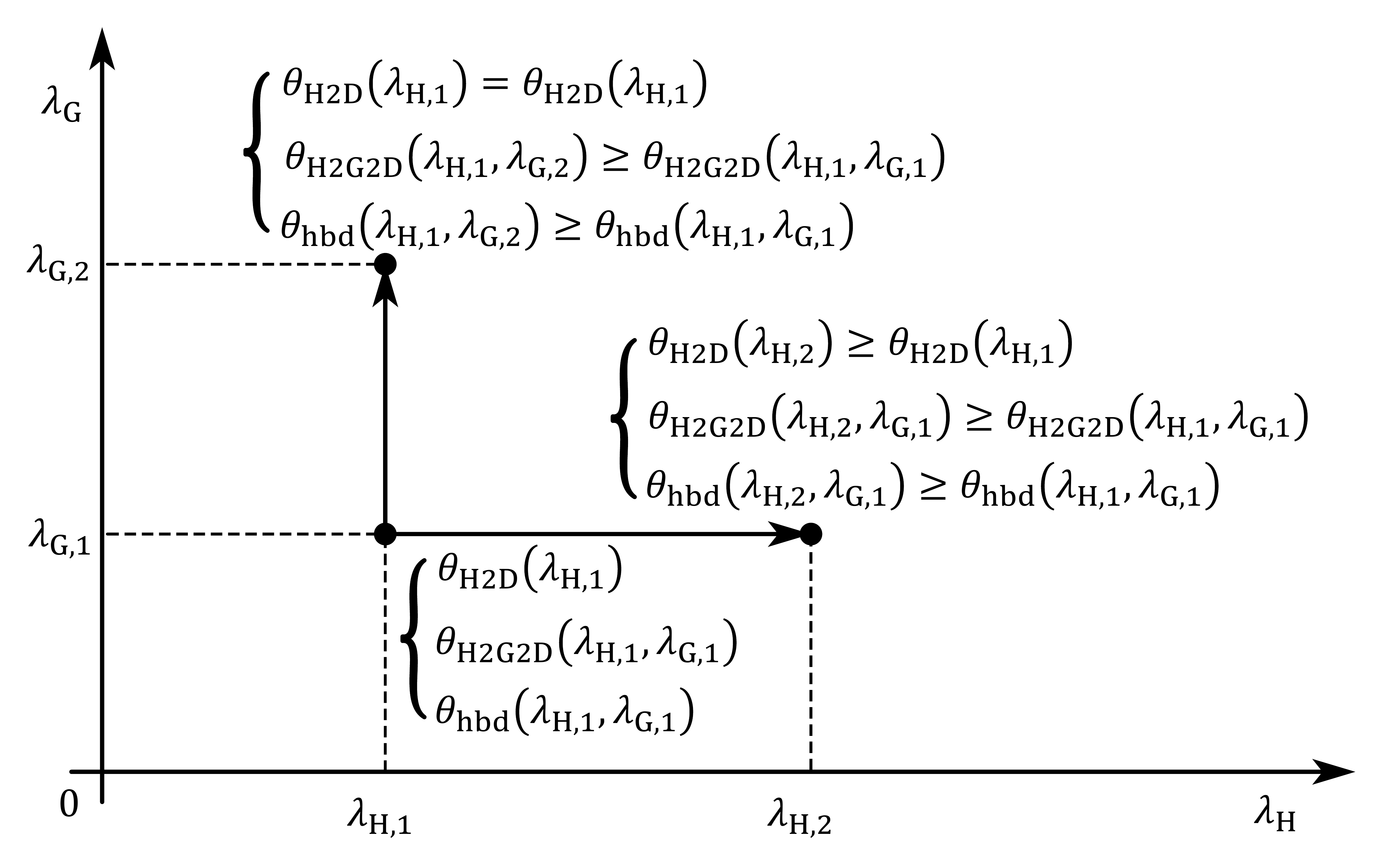}
    \caption{Illustration for the relationship between percolation probability and HAPS density $\lambda_{\rm H}$ and GW density $\lambda_{\rm G}$ in each coverage schemes.}
    \label{fig:increaseH2G2D}
\end{figure}
\section{Proof of Lemma \ref{lem:criticalH2G2D}}\label{app:criticalH2G2D}
Different from the single layer coverage model like the H2D coverage model, the critical condition for phase transition of percolation probability can be affected by the HAPS density $\lambda_{\rm H}$ and GW density $\lambda_{\rm G}$ because it can be formulated as a non-decreasing function of each of them. The critical condition infers that the increase in $\lambda_{\rm H}$ or $\lambda_{\rm G}$ will lead to a non-zero percolation probability. Therefore, there are three cases when $\theta_{\rm H2G2D}(\lambda_{\rm H,1},\lambda_{\rm G,1})=0$, $\lambda_{\rm H,2}=\lambda_{\rm H,1}+\Delta \lambda_{\rm H}$ and $\lambda_{\rm G,2}=\lambda_{\rm G,1}+\Delta\lambda_{\rm G}$ where for all $\Delta\lambda_{\rm H}>0$ and $\Delta\lambda_{\rm G}>0$: 
\begin{equation}
    \left.
\begin{array}{rr@{}l}
    \rmnuma:&
    \left\{\begin{array}{l}
    \theta_{\rm H2G2D}(\lambda_{\rm H,2},\lambda_{\rm G,1})>0 \\ \theta_{\rm H2G2D}(\lambda_{\rm H,1},\lambda_{\rm G,2})>0 \end{array}\right.,  \\
    \rmnumb:&
    \left\{\begin{array}{l}
    \theta_{\rm H2G2D}(\lambda_{\rm H,2},\lambda_{\rm G,1})>0 \\ \theta_{\rm H2G2D}(\lambda_{\rm H,1},\lambda_{\rm G,2})=0 \end{array}\right.,  \\
    \rmnumc:&
    \left\{\begin{array}{l}
    \theta_{\rm H2G2D}(\lambda_{\rm H,2},\lambda_{\rm G,1})=0 \\ \theta_{\rm H2G2D}(\lambda_{\rm H,1},\lambda_{\rm G,2})>0 \end{array}\right..  \\
\end{array}
    \right.
\end{equation}
It is worth noting that, if $\theta_{\rm H2G2D}(\lambda_{\rm H,2},\lambda_{\rm G,1})=0$ and $\theta_{\rm H2G2D}(\lambda_{\rm H,1},\lambda_{\rm G,2})=0$, the point $(\lambda_{\rm H,1},\lambda_{\rm G,1})$ does not satisfy the optimization object in (\ref{perproH2G2DHAPS}) and (\ref{perproH2G2DGW}).\\
\indent We first consider the case $\rmnuma$. Because $\theta_{\rm H2G2D}(\lambda_{\rm H,1},\lambda_{\rm G,1})=0$ and $\theta_{\rm H2G2D}(\lambda_{\rm H,2},\lambda_{\rm G,1})>0$, there exist $\lambda_{\rm H}'$ where $\lambda_{\rm H,1}\leq \lambda_{\rm H}'<\lambda_{\rm H,2}$ and 
\begin{equation}
\begin{array}{rl}
     \theta_{\rm H2G2D}(\lambda_{\rm H},\lambda_{\rm G,2})=0&\lambda_{\rm H}\leq\lambda_{\rm H}',  \\
     \theta_{\rm H2G2D}(\lambda_{\rm H},\lambda_{\rm G,2})>0&\lambda_{\rm H}>\lambda_{\rm H}'.
\end{array}
\end{equation}
Similarly, because $\theta_{\rm H2G2D}(\lambda_{\rm H,1},\lambda_{\rm G,1})=0$ and $\theta_{\rm H2G2D}(\lambda_{\rm H,1},\lambda_{\rm G,2})>0$, there exist $\lambda_{\rm G}'$ where $\lambda_{\rm G}^{-}\leq\lambda_{\rm G,1}\leq \lambda_{\rm G}'<\lambda_{\rm G,1}$ and 
\begin{equation}
\begin{array}{rl}
     \theta_{\rm H2G2D}(\lambda_{\rm H,2},\lambda_{\rm G})=0&\lambda_{\rm G}<\lambda_{\rm G}',  \\
     \theta_{\rm H2G2D}(\lambda_{\rm H,2},\lambda_{\rm G})>0&\lambda_{\rm G}>\lambda_{\rm G}'. 
\end{array}
\end{equation}
Therefore, the points $(\lambda_{\rm H}',\lambda_{\rm G,2})$ and $(\lambda_{\rm H,2},\lambda_{\rm G}')$ both satisfy the critical condition where $\lambda_{\rm H}'<\lambda_{\rm H,2}$ and $\lambda_{\rm G,2}>\lambda_{\rm G}'$. For any $\Delta\lambda_{\rm H}>0$ and $\Delta\lambda_{\rm G}>0$, such a relationship holds. Therefore, if we express the critical value of $\lambda_{\rm G}$ as a function of $\lambda_{\rm H}$, this function should be decreasing, vice versa.\\
\indent Next, we discuss the case $\rmnumb$. 
In this case, the increase in the $\lambda_{\rm G}$ does not lead to a non-zero percolation probability, that is $\theta_{\rm H2G2D}(\lambda_{\rm H,1},\lambda_{\rm G,2})=0$. For any $\lambda_{\rm H,2}>\lambda_{\rm H,1}$, $\theta_{\rm H2G2D}(\lambda_{\rm H,2},\lambda_{\rm G,2})\geq \theta_{\rm H2G2D}(\lambda_{\rm H,2},\lambda_{\rm G,1})$, so that $\theta_{\rm H2G2D}(\lambda_{\rm H,2},\lambda_{\rm G,2})>0$ for all $\lambda_{\rm H,2}>\lambda_{\rm H,1}$. For any $\lambda_{\rm H}\leq\lambda_{\rm H,1}$, $\theta_{\rm H2G2D}(\lambda_{\rm H},\lambda_{\rm G,2})\leq \theta_{\rm H2G2D}(\lambda_{\rm H,1},\lambda_{\rm G,2})$ and $\theta_{\rm H2G2D}(\lambda_{\rm H},\lambda_{\rm G,2})=0$. Therefore, all points $(\lambda_{\rm H,1},\lambda_{\rm G})$ where $\lambda_{\rm G}\geq \lambda_{\rm G,1}$ satisfy the critical condition. For a relatively high $\lambda_{\rm G}$, the requirement of $\lambda_{\rm H}$ for phase transition approaches to a constant. \\
\indent Similarly, in the case $\rmnumc$, the increase in $\lambda_{\rm H}$ does not lead to a non-zero percolation probability, that is $\theta_{\rm H2G2D}(\lambda_{\rm H,2},\lambda_{\rm G,1})=0$ for any $\lambda_{\rm H,2}>\lambda_{\rm H,1}$. For any $\lambda_{\rm G,2}>\lambda_{\rm G,1}$, $\theta_{\rm H2G2D}(\lambda_{\rm H,2},\lambda_{\rm G,2})\geq \theta_{\rm H2G2D}(\lambda_{\rm H,1},\lambda_{\rm G,2})$, so that $\theta_{\rm H2G2D}(\lambda_{\rm H,2},\lambda_{\rm G,2})>0$. For any $\lambda_{\rm G}^{-}\leq\lambda_{\rm G}\leq\lambda_{\rm G,1}$, $\theta_{\rm H2G2D}(\lambda_{\rm H,2},\lambda_{\rm G})\leq \theta_{\rm H2G2D}(\lambda_{\rm H,2},\lambda_{\rm G,1})$ and $\theta_{\rm H2G2D}(\lambda_{\rm H,2},\lambda_{\rm G})=0$. Therefore, all points $(\lambda_{\rm H},\lambda_{\rm G,1})$ where $\lambda_{\rm H}\geq \lambda_{\rm H,1}$ satisfy the critical condition. For a relatively high $\lambda_{\rm H}$, the requirement of $\lambda_{\rm G}$ for phase transition approaches to a constant.\\
\indent In summary, the case $\rmnumb$ shows that the critical value of $\lambda_{\rm H}$ has a lower bound and approaches to a constant when $\lambda_{\rm G}$ goes to infinity. The case $\rmnumc$ shows that the critical value of $\lambda_{\rm G}$ has a lower bound and approaches to a constant when $\lambda_{\rm H}$ goes to infinity. The case $\rmnuma$ describes the critical condition between these two cases, where the critical value of HAPS density and GW density can be non-increasing functions of each other. We can represent the set of all points in these three cases using an implicit function $\mathcal{F}(\lambda_{\rm H},\lambda_{\rm G})=0$. \\
\indent It is worth noting that, the lower bound of $\lambda_{\rm H}$ for $\mathcal{F}(\lambda_{\rm H},\lambda_{\rm G})=0$ is between 0 and $\lambda_{\rm H}^{*}$, and the lower bound of $\lambda_{\rm G}$ for $\mathcal{F}(\lambda_{\rm H},\lambda_{\rm G})=0$ is between $\lambda_{\rm G}^{-}$ and $\lambda_{\rm G}^{+}$. This is because the super-critical region is located in the area $\lambda_{\rm H}>\lambda_{\rm H}^{*}$ and $\lambda_{\rm G}>\lambda_{\rm G}^{+}$, and the sub-critical region is located in the area $\lambda_{\rm G}<\lambda_{\rm G}^{-}$.

\section{Proof of Lemma \ref{lem:increaseHybrid}}\label{app:increaseHybrid}
For a fixed value of HAPS density $\lambda_{\rm H,1}$, we consider the same set of HAPSs $\Psi$ with density $\lambda_{\rm H,1}$, and two different sets of GWs $\Phi_1$ and $\Phi_2$ with densities $\lambda_{\rm G,1}$ and $\lambda_{\rm G,2}$, respectively, where $0<\lambda_{\rm G,1}<\lambda_{\rm G,2}$. Since $\Phi_1$ and $\Phi_2$ are both PPPs, $\Phi_1$ can be constructed by thinning $\Phi_2$ with probability $\frac{\lambda_{\rm G,1}}{\lambda_{\rm G,2}}$ and $\Phi_1\subseteq \Phi_2$. The set $\Phi_2\backslash\Phi_1$ can be considered a PPP with density $\lambda_{\rm G,2}-\lambda_{\rm G,1}$. The same as the H2G2D case, appending any point in $\Phi_2\backslash\Phi_1$ into $\Phi_1$ does not remove any point in $\Phi_c$, so that $\Phi_{c,1}\subseteq\Phi_{c,2}$ and $V_{\rm H2G2D,1}\subseteq V_{\rm H2G2D,2}$. With the same group of HAPSs, we also have $V_{\rm H2D,1}=V_{\rm H2D,2}=\Psi$. Since $V_{\rm hbd}=V_{\rm H2D}\cup V_{\rm H2G2D}=\Phi_{c}\cup\Psi$, we have $V_{\rm hbd,1}\subseteq V_{\rm hbd,2}$. We know that $E_{\rm hbd}=E_{\rm H2D}\cup E_{\rm H2G2D}\cup E_{\rm HG}$, appending points in $\Phi_{c,1}$ leads to the extension of $E_{\rm H2G2D}$ and $E_{\rm HG}$, that is $E_{\rm H2G2D,1}\subseteq E_{\rm H2G2D,2}$ and $E_{\rm HG,1}\subseteq E_{\rm HG,2}$. The edge set that describes the connections between HAPSs' direct coverage does not change, \ie $E_{\rm H2D,1}=E_{\rm H2D,2}$. Therefore, we have $E_{\rm hbd,1}\subseteq E_{\rm hbd,2}$, and then $G_{\rm hbd,1}\subseteq G_{\rm hbd,2}$. The giant components also satisfy $K_{\rm hbd,1}(0)\subseteq K_{\rm hbd,2}(0)$ and $|K_{\rm hbd,1}(0)|\leq |K_{\rm hbd,2}(0)|$, and percolation probabilities satisfy:
\begin{equation}
    \theta_{\rm hbd}(\lambda_{\rm H,1},\lambda_{\rm G,2})\geq \theta_{\rm hbd}(\lambda_{\rm H,1},\lambda_{\rm G,1}),
\end{equation}
and percolation probability in the hybrid coverage scheme is a non-decreasing function of $\lambda_{\rm G}$.

Next, for a fixed value of GW density $\lambda_{\rm G,1}$, we consider the same set of GWs $\Phi$ with density $\lambda_{\rm G,1}$, and two different sets of HAPSs $\Psi_1$ and $\Psi_2$ with densities $\lambda_{\rm H,1}$ and $\lambda_{\rm H,2}$, respectively, where $0<\lambda_{\rm H,1}<\lambda_{\rm H,2}$. Since $\Psi_1$ and $\Psi_2$ are both PPPs, $\Psi_1$ can be constructed by thinning $\Psi_2$ with probability $\frac{\lambda_{\rm H,1}}{\lambda_{\rm H,2}}$ and $\Psi_1\subseteq \Psi_2$. The set $\Psi_2\backslash \Psi_1$ can be considered a PPP with density $\lambda_{\rm H,2}-\lambda_{\rm H,1}$. The same as the H2G2D coverage scheme, appending points in $\Psi_2\backslash\Psi_1$ into $\Psi_1$ makes more GWs connected to the core network, so that $\Phi_{c,1}\subseteq\Phi_{c,2}$ and $V_{\rm H2G2D,1}\subseteq V_{\rm H2G2D,2}$. Because $\Psi_1\subseteq\Psi_2$, we have $V_{\rm H2D,1}\subseteq V_{\rm H2D,2}$. Since $V_{\rm hbd}=V_{\rm H2D}\cup V_{\rm H2G2D}=\Phi_{c}\cup\Psi$, we have $V_{\rm hbd,1}\subseteq V_{\rm hbd,2}$.  At the same time, due to the increase in number of HAPSs and connected GWs, the edge sets satisfy:
$E_{\rm H2D,1}\subseteq E_{\rm H2D,2}$, $E_{\rm H2G2D,1}\subseteq E_{\rm H2G2D,2}$ and $E_{\rm HG,1}\subseteq E_{\rm HG,2}$, therefore $E_{\rm hbd,1}\subseteq E_{\rm hbd,2}$. The random graphs $G_{\rm hbd,1}\subseteq G_{\rm hbd,2}$ and the giant components satisfy $K_{\rm hbd,1}(0)\subseteq K_{\rm hbd,2}(0)$ and $|K_{\rm hbd,1}(0)|\geq |K_{\rm hbd,2}(0)|$. The percolation probabilities satisfy:
\begin{equation}
    \theta_{\rm hbd}(\lambda_{\rm H,2},\lambda_{\rm G,1})\geq \theta_{\rm hbd}(\lambda_{\rm H,1},\lambda_{\rm G,1}),
\end{equation}
and percolation probability in the hybrid coverage scheme is a non-decreasing function of $\lambda_{\rm H}$.

\section{Proof of Lemma \ref{lem:criticalHybrid}}\label{app:criticalhybrid}
In Case 1, the supercritical region is on the right side of $\lambda_{\rm H}=\lambda_{\rm H}^{+}$, and it does not have overlapping areas with the subcritical region. In Case 2, the supercritical region is on the right side of $\lambda_{\rm H}=\lambda_{\rm H}^{*}$ and does not overlap with the subcritical region. In Case 3, we need to prove that 
\begin{equation}
\begin{array}{r@{}l}
\frac{1}{\pi (r_{\rm G2D}-a)^2} \ln (1+\frac{e^{-\lambda_{\rm H}\pi(r_{\rm H2D}-a)^2}-\frac{1}{2}}{\frac{1}{2}-e^{-\lambda_{\rm H}\pi (r_{\rm H2G}-r_{\rm G2D}+a)^2}})\\
\geq\frac{\ln 2}{\pi (r_{\rm G2D}+a)^2}-\lambda_{\rm H}\frac{(r_{\rm H2D}+a)^2}{(r_{\rm G2D}+a)^2}
\end{array}
\end{equation}
when $\lambda_{\rm H}^{*}<\lambda_{\rm H}<\lambda_{\rm H}^{-}$. We can notice that, when $\lambda_{\rm H}^{*}<\lambda_{\rm H}<\lambda_{\rm H}^{+}$, we have
\begin{equation}
\begin{array}{r@{}l}
    \frac{\ln 2}{\pi (r_{\rm G2D}+a)^2}&-\lambda_{\rm H}\frac{(r_{\rm H2D}+a)^2}{(r_{\rm G2D}+a)^2}\\
    &<\frac{\ln 2}{\pi (r_{\rm G2D}-a)^2}-\lambda_{\rm H}\frac{(r_{\rm H2D}-a)^2}{(r_{\rm G2D}-a)^2}
\end{array}
\end{equation}
and 
\begin{equation}
\begin{array}{r@{}l}
    &\frac{1}{\pi (r_{\rm G2D}-a)^2} \ln (1+\frac{e^{-\lambda_{\rm H}\pi(r_{\rm H2D}-a)^2}-\frac{1}{2}}{\frac{1}{2}-e^{-\lambda_{\rm H}\pi (r_{\rm H2G}-r_{\rm G2D}+a)^2}})\\
    \geq&\frac{1}{\pi (r_{\rm G2D}-a)^2} \ln (1+\frac{e^{-\lambda_{\rm H}\pi(r_{\rm H2D}-a)^2}-\frac{1}{2}}{\frac{1}{2}})\\
    =&\frac{1}{\pi (r_{\rm G2D}-a)^2} \ln (2e^{-\lambda_{\rm H}\pi(r_{\rm H2D}-a)^2})\\
    =&\frac{\ln 2}{\pi (r_{\rm G2D}-a)^2}-\lambda_{\rm H}\frac{(r_{\rm H2D}-a)^2}{(r_{\rm G2D}-a)^2}.
\end{array}
\end{equation}

Therefore, the supercritical region and subcritical region do not overlap. Similar to the H2G2D coverage scheme, the percolation probability $\theta_{\rm hbd}$ is a non-decreasing function about $\lambda_{\rm H}$ or $\lambda_{\rm G}$. Consider the curve where $\theta_{\rm hbd}(\lambda_{\rm H},\lambda_{\rm G})=p$ $\forall p>0$, the set of solutions also corresponds to a curve. In this curve, the increase in the $\lambda_{\rm H}$ does not lead to a higher value of $\lambda_{\rm G}$'s solution and the increase in $\lambda_{\rm G}$ does not lead to a higher value of $\lambda_{\rm H}$'s solution as well. When $p$ approaches 0, we can obtain the curve where the percolation probability meets the phase transition, which we defined as $\mathcal{G}(\lambda_{\rm H},\lambda_{\rm G})=0$. In the hybrid coverage scheme, the value of $\lambda_{\rm H}$ for $\mathcal{G}(\lambda_{\rm H},\lambda_{\rm G})=0$ should not exceed $\lambda_{\rm H}^{+}$ because $\theta_{\rm hbd}(\lambda_{\rm H},\lambda_{\rm G})>0$ in the super-critical case.

\ifCLASSOPTIONcaptionsoff
  \newpage
\fi

\bibliographystyle{IEEEtran}
\bibliography{ref}

\end{document}